\PassOptionsToPackage{usenames,dvipsnames}{xcolor}
\documentclass[11pt]{article}

\usepackage{amsmath,amsthm,mathtools,thmtools}
\usepackage{enumitem}

\usepackage[T1]{fontenc}
\usepackage{url}
\usepackage{nicefrac}

\usepackage{hyperref}
\hypersetup{
    colorlinks=true,
    linkcolor=blue,
    filecolor=magenta,      
    urlcolor=cyan,
    pdfpagemode=FullScreen,
    }

\usepackage{cleveref}
\usepackage{float}

\declaretheorem[numberwithin=section]{theorem}

\declaretheorem[numberlike=theorem,style=definition]{definition}
\declaretheorem[numberlike=theorem]{claim,observation,corollary,proposition,lemma}
\declaretheorem[numberlike=theorem,style=remark]{remark}

\Crefname{assumption}{Assumption}{Assumptions}
\declaretheorem[numberwithin=theorem,name=Claim]{claim-inside-theorem}
\Crefname{claim-inside-theorem}{Claim}{Claims}
\declaretheorem[numberwithin=theorem,name=Lemma]{lemma-inside-theorem}
\Crefname{lemma-inside-theorem}{Lemma}{Lemmas}
\declaretheorem[numberlike=claim-inside-theorem,name=Observation]{observation-inside-theorem}
\declaretheorem[numberwithin=lemma,name=Claim]{claim-inside-lemma}
\Crefname{claim-inside-lemma}{Claim}{Claims}
\declaretheorem[numberwithin=claim-inside-lemma,name=Subclaim]{subclaim-inside-claim-inside-lemma}
\Crefname{subclaim-inside-claim-inside-lemma}{Subclaim}{Subclaims}
\declaretheorem[numberlike=claim-inside-lemma,name=Corollary]{corollary-inside-lemma}
\declaretheorem[numberwithin=lemma,name=Lemma]{lemma-inside-lemma}
\declaretheorem[numberlike=claim-inside-lemma,name=Observation]{observation-inside-lemma}
\declaretheorem[numberwithin=lemma-inside-theorem,name=Claim]{claim-inside-lemma-inside-theorem}
\declaretheorem[numberwithin=lemma-inside-theorem,name=Lemma]{lemma-inside-lemma-inside-theorem}
\declaretheorem[numberwithin=lemma-inside-theorem,name=Observation]{observation-inside-lemma-inside-theorem}
\crefname{claim}{claim}{claims}
\Crefname{algocf}{Algorithm}{Algorithms}

\newcounter{personalcounterzero}
\newcounter{personalcounterone}
\newcounter{personalcountertwo}
\newcounter{personalcounterthree}

\Crefname{paragraph}{Paragraph}{Paragraphs}

\newcommand{\C}{\mathcal{C}} 

\renewcommand{\epsilon}{\varepsilon}

\usepackage[vlined,ruled]{algorithm2e}

\SetKw{AAnd}{and}
\SetKw{Or}{or}
\SetKw{Not}{not}
\SetKw{Break}{break}
\SetAlFnt{\small}
\SetAlCapFnt{\small}
\SetAlCapNameFnt{\small}
\SetAlCapHSkip{-1ex}
\IncMargin{-\parindent}
\SetAlgoCaptionSeparator{}

\usepackage[edges]{forest}		
\usepackage{tikz}
\usetikzlibrary{arrows.meta}
\usepackage{tikz-3dplot}
\usetikzlibrary{calc,arrows,automata,patterns,graphs,shapes,petri,decorations.pathmorphing,decorations.markings,decorations.pathreplacing} 
\tikzset{->-/.style={decoration={markings,mark=at position .5 with {\arrow{>}}},postaction={decorate}}}
\tikzset{vertex/.style={draw,circle,inner sep=0pt,minimum size=18pt},>=latex'}
\tikzset{c/.style={draw,circle,inner sep=0pt,minimum size=15pt},>=latex'}
\tikzset{dot/.style={draw,circle,inner sep=0pt,minimum size=3pt, fill=black},>=latex'}

\usepackage{xcolor}
\definecolor{lightblue}{rgb}{0.38,0.82,0.90}
\definecolor{darkgreen}{rgb}{0.2,0.8,0.55}

\usepackage[colorinlistoftodos]{todonotes}

\newcommand{\hide}[1]{}

\newcommand{\shortv}[1]{\hide{#1}}

\newcommand{\vnar}{\textup{\textsc{VnaR}}\xspace}

\newcommand{\hampath}{\textup{\textsc{Hamiltonian Path}}\xspace}
\newcommand{\tourn}{\mathcal{T}}
\newcommand{\pop}{\mathrm{pop}}

\newcommand{\beats}{\rightarrow}
\newcommand{\nph}{\textup{\textsf{NP-hard}}\xspace}
\newcommand{\vcprob}{\ensuremath{\mathrm{VC_{prob}}}\xspace}
\newcommand{\score}{\ensuremath{\mathrm{score}}\xspace}
\newcommand{\bigoh}{\mathcal{O}}

\newcommand{\FPT}{\textup{FPT}}
\newcommand{\XP}{\textup{XP}}

\newcommand{\vect}{\textsf{vec}}

\newcommand{\opt}{\textsf{OPT}\xspace}

\usepackage{graphicx}

\usepackage{lineno}

\newcommand{\etal}{et al.\xspace}

\newcommand{\ctc}{CtC\xspace}

\usepackage[margin=1in]{geometry}
\usepackage{charter}
\usepackage[english]{babel}
\usepackage{setspace}
\usepackage{tcolorbox}
\usepackage{amsmath,amssymb}

\makeatletter
\def\algbackskip{\hskip-\ALG@thistlm}
\makeatother

\title{Robust Value Maximization in Challenge the Champ Tournaments with Probabilistic Outcomes}
\author{
    Umang Bhaskar
    \thanks{Tata Institute of Fundamental Research, Mumbai. {\tt umang@tifr.res.in}} \and 
    Juhi Chaudhary 
    \thanks{Indian Institute of Petroleum and Energy, Visakhapatnam. {\tt juhi.math@iipe.ac.in}} \and 
    Sushmita Gupta 
    \thanks{The Institute of Mathematical Sciences, HBNI, Chennai. {\tt sushmitagupta@imsc.res.in}} \and 
    Pallavi Jain 
    \thanks{Indian Institute of Technology Jodhpur, Jodhpur. {\tt pallavi@iitj.ac.in}} \and 
    Sanjay Seetharaman
    \thanks{The Institute of Mathematical Sciences, HBNI, Chennai. {\tt sanjays@imsc.res.in}}
 }
\date{}

\allowdisplaybreaks

\begin{document}

\maketitle

\begin{abstract}
Challenge the Champ is a simple tournament format, where an ordering of the players --- called a seeding --- is decided. The first player in this order is the initial champ, and faces the next player. The outcome of each match decides the current champion, who faces the next player in the order. Each player also has a popularity, and the value of each match is the popularity of the winner. Value maximization in tournaments has been previously studied when each match has a deterministic outcome. However, match outcomes are often probabilistic, rather than deterministic. We study robust value maximization in Challenge the Champ tournaments, when the winner of a match may be probabilistic.
That is, we seek to maximize the total value that is obtained, irrespective of the outcome of probabilistic matches. 
We show that even in simple binary settings, for non-adaptive algorithms, the optimal robust value --- which we term the \textsc{VnaR}, or the value not at risk --- is hard to approximate. 
However, if we allow adaptive algorithms that determine the order of challengers based on the outcomes of previous matches, or restrict the matches with probabilistic outcomes, we can obtain good approximations to the optimal \textsc{VnaR}. 
\end{abstract}

\thispagestyle{empty}
\newpage
\setcounter{page}{1}

\section{Introduction}

\begin{sloppypar}
Tournaments are a fundamental mechanism for structuring competition. Beyond their well-established role in organizing play in sports and e-sports, where both competitive fairness and entertainment value matter, they also arise in other decision-making contexts, such as elections, where pairwise comparisons are meaningful~\cite{vu2009complexity,buchanan1980toward,rosen1985prizes,laslier1997tournament}. 
\end{sloppypar}

In tournaments, some matches naturally generate more audience interest than others --- either because they feature a marquee player who is highly popular or because the matchup, such as between arch-rivals, is anticipated to be dramatic. Organizers of sporting events, as well as competitive platforms more broadly, often try to structure play to maximize such engagement value, and in turn the revenue. An especially simple yet versatile format is the Challenge the Champ (CtC) tournament: a fixed “champion seat” is defended against a sequence of challengers whose order --- the {\it seeding} --- is chosen by the organizer. The seeding strongly shapes the tournament’s trajectory, and thus, by choosing the seeding, the organizer significantly influences the excitement it generates. 

Formally, a \ctc tournament, is a permutation of the players, called the seeding. The first player is the initial champion and is challenged by the second player. The winner of this match then plays the third player, and so on. With $n$ players, a winner (i.e., the final champion) is determined in $n-1$ matches. \ctc is a specific instance of single-elimination tournaments, historically used in boxing, martial arts, and other combat sports where challengers rise to take on the current champion (or master).

A major goal of tournament organizers is to maximize the value --- this could be the revenue, total viewership, tickets sold, etc. While there are many tournament value functions studied, typically each possible match has a value, and the value of the tournament is the sum of the values of the matches played. If outcomes were perfectly predictable, this design question would then be easy to state: choose a seeding that maximizes the total value of the tournament. 
In real-world scenarios, some pairwise results may be essentially certain, perhaps due to a large skill gap, while others are inherently uncertain and are modeled by probabilities. This probabilistic setting is a natural way to model real-life competitions, a fact that was noted in earlier papers~\cite{vu2009complexity,aziz2014fixing,BlazejG0S24}. 
Existing works on designing \ctc tournaments, and more general ones, have either assumed full determinism~\cite{bhaskar2025maximizing} or aimed to optimize expectation-based metrics such as expected wins or revenue~\cite{VuShoham2011}. But expectation alone can be misleading: a seeding that is strong on average may fail to produce enough, or any, exciting games if a few upsets occur. This motivates the search for a \emph{robust guarantee} in tournament design.\footnote{This is different from the notion of robustness discussed in \cite{ChatterjeeIbsenJensenTkadlec2016}, where it pertains to the error in the probability values.}

In this paper, we introduce a risk-averse alternative, called the {\sc Value Not At Risk} (\vnar), defined as the minimum total value an algorithm can obtain, over all realizations of the uncertain matches. The \vnar measures the value the organizer can guarantee, {\it no matter how the uncertain games turn out}. This parallels risk-aware reasoning in finance, where measures such as {\sc Value At Risk} ({\sc VaR}) quantify the guaranteed part of a portfolio’s return: instead of maximizing average value, the designer secures a safe baseline excitement level even under unfavorable outcomes.

Our contribution in this work is two-fold: Conceptually, we propose \vnar as a natural, risk-averse tournament design objective for the realistic setting where game outcomes are not fully predictable.
Algorithmically, we present the first study of \vnar-maximization in \ctc tournaments, and present a comprehensive set of results. For non-adaptive algorithms, we first establish the computational hardness of the problem even under severe restrictions, including additive and multiplicative approximation hardness and propose almost matching polynomial-time approximation algorithms. We then study adaptive algorithms, where decisions regarding future matches are deferred until earlier game outcomes are known, and show that good approximations can indeed be obtained. We note that adaptive tournament formats are used in real-world tournaments also, such as the FIDE Grand Swiss Tournament, where the contestants are paired to ensure that each player plays an opponent with a similar running score without playing the same opponent more than once. Match pairing for each round is done after the previous round has ended and depends on its results.

In the past, \ctc tournaments, also known as \emph{stepladder} tournaments in prior work, have been studied from multiple perspectives. These include satisfying axiomatic notions of fairness~\cite{ArlegiD20} and characterizing strength graphs under which a designated favorite player can win the tournament~\cite{knockout,YangD21}. When the strength graphs are probabilistic, Mattei et al.~\cite{mattei2015complexity} studied how players could be bribed to lower their chances of beating the initial champion, aiming to maximize the champion’s overall probability of winning while staying within a given budget. Building on this, Chaudhary et al.~\cite{chaudhary2025parameterized} explored the problem further from the perspective of parameterized complexity. In this paper, we study adaptivity in the presence of probabilistic events. It is worth noting that adaptivity in \emph{knockout tournaments}, a form of single-elimination tournament, has been examined recently by Chaudhary~\etal\cite{chaudhary2025adaptive}, though primarily from the perspective of coalition manipulation. In that setting, coalition players can strategically decide in each round which games to throw, depending on the outcomes of the previous matches. To the best of our knowledge, this is the first work on probabilistic \ctc tournaments, where the goal is to maximize an objective function that quantifies the number of valuable games.

\paragraph{Related Work.} Starting with the influential work of Vu \etal~\cite{vu2009complexity}, a number of tournament formats have been widely adopted and systematically studied, including the \emph{knockout} format \cite{suksompong2021tournaments,williams_moulin_2016,blavzej2024controlling,efremenko2025tournament}, the \emph{round-robin} format \cite{russell2009manipulating,suksompong2016scheduling}, the \emph{double-elimination} format \cite{aziz2018fixing,stanton2013structure}, and the \emph{Swiss system} \cite{sauer2024improving,sziklai2022efficacy}. Each of these formats represents a distinct organizational principle, offering particular advantages and limitations in terms of ranking outcomes and elimination procedures. We note that across these various tournament models\hide{that are studied in computational social choice,} the goal is often to determine if there is a seeding that allows a particular player to win, or to study the effects of bribery, winner determination, etc. Quite naturally, one of the primary goals of tournaments is to maximize revenue, generated either through advertising, sale of tickets, sale of broadcast rights, and so on. This aspect of tournaments is studied in prior work on knockout tournaments ~\cite{chaudhary2024make, gupta2024exercise, dagaev2018competitive} and is now being studied for \ctc tournaments as well~\cite{bhaskar2025maximizing}.

Vu \etal\cite{vu2009complexity} define the \emph{probabilistic tournament fixing} question, where the input consists of a set of players, of which one is distinguished and the probability value $p_{ij}$ associated with a game involving players $i$ and $j$. The paper presented a recursive formula to calculate the winning probability of a player, and for the case of a balanced tournament, showed NP-hardness of deciding whether the distinguished player wins with at least some probability, even under additional restrictions. Subsequently, Aziz \etal\cite{aziz2014fixing} and Blazej \etal\cite{BlazejG0S24} have furthered the study of probabilistic tournament fixing in the knockout format.

\paragraph{Our Setting.} We study tournament value maximization when the outcome of matches can be uncertain. We focus on \ctc tournaments when the value of each match is either 1 or 0, determined by the popularity of the winner (this is termed \emph{player-popularity-based} tournament values). Our input thus consists of a popularity (1 or 0) for each player, as well as a strength graph where additionally each directed edge, say $(p_1,p_2)$, has a weight in $[0,1]$, indicating the probability of the outcome determined by the edge direction. Thus, if edge $(i,j)$ has weight $p_{ij} \in [0,1]$, then in a match between these players, player $i$ beats player $j$ with probability $p_{ij}$, while $j$ beats $i$ with probability $1-p_{ij}$. If the weight on an edge is 1, such a match (and the corresponding edge) is called \emph{deterministic}, while a weight strictly less than 1 is an \emph{uncertain} edge and the resulting outcome is said to be probabilistic. We say a player $i$ {\it beats} $j$ (or $j$ is beaten by $i$) to mean that $i$ beats $j$ deterministically, i.e., $p_{ij} = 1$. 

Define $\mathcal{P}$ as the set of popular players, $\mathcal{U}$ as the set of unpopular players that are beaten\hide{ deterministically} by at least one popular player, and $\mathcal{W}$ be the set of remaining unpopular players. Let $n_p = |\mathcal{P}|$ and $n_u = |\mathcal{U}|$. Prior work by Bhaskar~\etal\cite{bhaskar2025maximizing} shows that in \ctc tournaments, in the absence of uncertainty (i.e., when every match is deterministic), the optimal seeding --- the sequence of matches that maximizes the value --- can be obtained in polynomial time. In fact, if $n_p$ is the number of players with value 1 (termed \emph{popular players}) and $n_u$ is the number of players with value 0 (termed \emph{unpopular players}) that are beaten by at least one popular player, then the optimal tournament value is exactly $n_p + n_u -1$. 

\paragraph{Value not at Risk.} In this work, where some matches are uncertain, we focus on maximizing the total value obtained  (equivalently, the number of matches won by popular players) 
in the \emph{worst case with probability 1}. Our objective is thus to maximize the total value for a risk-averse tournament organizer. 
Formally stated, for a seeding $\sigma$, the \emph{value not at risk} (the $\vnar(\sigma$)) is the total value obtained with probability 1. That is, for any possible outcomes of uncertain matches, the total value obtained is at least the \vnar. Our primary objective is to obtain a seeding with maximum possible \vnar. Note that this may require uncertain matches to be played. See \Cref{fig:example-with-pu-prob-edges} for an illustration. 
The following proposition shows that, as in the deterministic setting, the unpopular players in $\mathcal{W}$ do not contribute to the optimal value.

\begin{restatable}{proposition}{optimalvnar}
\label{prop:optimalvnar}
    The optimal \vnar obtainable by any algorithm is at most $n_p + n_u - 1$.
\end{restatable}
\begin{proof}
   Note that, by definition, players in $\mathcal{W}$ are not (deterministically) beaten by any player in $\mathcal{P}$. Hence, for any algorithm (adaptive or non-adaptive), consider the total value obtained when all players in $\mathcal{W}$ beat all players in $\mathcal{P}$. Clearly, in any match with a player in $\mathcal{W}$, an unpopular player must win the match, contributing value 0. Hence, the total value obtainable is at most $n_p + n_u - 1$, with is an upper bound on the \vnar.
\end{proof}

If $\vnar^*$ denotes the optimal \vnar of a given instance, we say an algorithm is an $f$-additive approximation if the value obtained for any outcome is at least $\vnar^*-f$. Similarly, an algorithm is an $f$-multiplicative approximation if the value obtained for any outcome is at least $f \cdot \vnar^*$.

\paragraph{Our Contribution.} We first note that a non-adaptive seeding with \vnar $n_p + \lceil \frac{n_u}{n_p}\rceil -1$ is easily obtained as follows. Let $a_p$ be a popular player that\hide{ deterministically} beats the maximum number of unpopular players in $\mathcal{U}$, say $k$. Clearly, $k \ge \lceil n_u / n_p \rceil$. In the seeding, we first put $a_p$, followed by the $k$ unpopular players\hide{ deterministically} beaten by $a_p$. This is followed by the remaining popular players, and finally, all the other players. Under this seeding, the first $n_p+k-1$ matches are all won by popular players, yielding a value of at least $n_p + \frac{n_u}{n_p}-1$. From~\Cref{prop:optimalvnar}, if $n_u \le c \cdot n_p$, a multiplicative $\max \{1/n_p, 1/(c+1)\}$-approximate non-adaptive seeding is thus readily obtainable.

We start with studying non-adaptive algorithms that output a \emph{fixed seeding} without seeing any probabilistic outcomes. Even in this simple setting of binary values, we show strong negative results. In fact, our negative results hold when the optimal \vnar is $n_p + n_u - 1$, the maximum possible. 
These results imply that computing a fixed seeding that obtains a large value with probability 1 is hard, even when the optimal \vnar is actually the largest possible value, $n_p+n_u - 1$.

Our first hardness result is the following.
\begin{restatable}{theorem}{thmadditivepnp}
    \label{thm:additive approximation NP hardness}
    Unless $\textup{\textsf{P=NP}}$, there is no non-adaptive polynomial-time algorithm that obtains a seeding that additively approximates the optimal \vnar with a factor of 
    $(n_p)^{1-\varepsilon}/(1+\alpha)$
    for any fixed $0 < \varepsilon \le 1$ and $\alpha>0$.
\end{restatable}

Under the Exponential-time Hypothesis (ETH) \cite{DBLP:journals/jcss/ImpagliazzoP01}, we have a stronger statement of hardness if $n_p^{\varepsilon} > (f(n_p)/2){\log^2n_p}$. 

\begin{restatable}{theorem}{thmc}\label{thm:VNAR is ETH-hard}
\label{thm:additive approximation ETH hardness}
\begin{sloppypar}
    Assuming the \textup{ETH}, there is no non-adaptive polynomial-time algorithm that obtains a seeding that additively approximates the optimal \vnar with a factor
    $\frac{n_p}{f(n_p)\log^2 n_p}$,
    for any polynomial-time computable and non-decreasing function $f \in \omega(1)$.
\end{sloppypar}
\end{restatable}

\begin{sloppypar}
Furthermore, we show that unless $\textup{\textsf{P=NP}}$, no non-adaptive polynomial-time algorithm can multiplicatively approximate the optimal \vnar by a factor better than $1/\sqrt{n_p}$. 
\end{sloppypar}

\begin{restatable}{theorem}{thma}
\label{thm:multiplicative-approximation-NP-hardness}
  
    Unless \textup{\textsf{P=NP}}, there is no non-adaptive polynomial-time algorithm that obtains a seeding that multiplicatively approximates the optimal \vnar with a factor
    $\frac{1}{n_p^{1-\varepsilon}}$ for any fixed $0.5 < \varepsilon \le 1$.
\end{restatable}
Let \opt denote the optimal \vnar of an instance. We note that \Cref{thm:additive approximation NP hardness} rules out a polynomial-time algorithm that obtains a seeding of \vnar at least $\opt- (n_p)^{1-a}/(1+\alpha)$, where $0< a\leq 1$, while \Cref{thm:multiplicative-approximation-NP-hardness} rules out one with \vnar at least $\opt/(n_p)^{1-b}$, where $0.5 < b \leq 1$. 
Thus, when $n_p$ is sufficiently large, the latter gives a stronger inapproximability guarantee than the former.

Given the negative results for non-adaptive algorithms, we then study adaptive algorithms that can obtain, by deferring the decision about who is introduced next, a large value with probability 1. Thus such an algorithm constructs a seeding adaptively, depending on the outcome of previous matches.
For adaptive algorithms, we are able to show an additive approximation to the optimal \vnar that depends on the \emph{chromatic number} of a certain graph.\footnote{ The \emph{chromatic number} of a graph is the minimum number of colors required to obtain a proper coloring of the graph, i.e., an assignment of colors to the vertices of such that no two adjacent vertices receive the same color.} 
In particular, let $G_p=(V_p,E_p)$ be the subgraph of the strength graph such that $V_p$ is the set of popular players and $E_p$ is the set of uncertain edges incident on $V_p$. 

\begin{restatable}{theorem}{thmb} \label{thm:chromatic}
\label{thm:pp-prob-edges-chrom-number}
Given a proper coloring of $G_p$ with $k$ colors, there is a polynomial-time adaptive algorithm that obtains $\vnar$ at least $(n_p+n_u-1)-(k-1)$.
\end{restatable}

We note that, in general the chromatic number is hard to approximate~\cite{KhannaLS00} but if the number of uncertain edges is small (say $q$), one can indeed compute the chromatic number of $G_p$ in time $O^*(4^q)$, where the $O^*$ notation hides polynomial factors~\cite{BjorklundHK09}.
Moreover, using a 2-approximation algorithm for Vertex Cover on $G_p$, we can obtain a seeding with \vnar at least $(n_p + n_u - 1) - (2\vcprob - 1)$ in polynomial time, where $\vcprob$ is the size of a minimum vertex cover in $G_p$. 
Similarly, we obtain polynomial-time algorithms when $G_p$ belongs to graph classes where an optimal coloring can be computed efficiently (e.g., perfect graphs).

By observing that $k \le n_p$, and combining with the deterministic seeding that obtains \vnar $n_p + \lceil n_u/n_p\rceil -1$, we obtain the following result, demonstrating the large gap between adaptive and non-adaptive algorithms for this problem.

\begin{sloppypar}
\begin{restatable}{theorem}{finaladaptive} \label{thm:finaladaptive}
There is a polynomial-time adaptive algorithm that obtains \vnar at least $\left(n_p + n_u - 1\right)/2$.
\end{restatable}
\end{sloppypar}
\begin{proof}
    We consider two cases. If $n_u \le n_p$, then the non-adaptive seeding that obtains value $n_p + \lceil n_u / n_p \rceil - 1$ is a multiplicative 2-approximation. Else, we run the adaptive algorithm from~\Cref{thm:chromatic}, with the chromatic number upper bounded by $n_p$, i.e., each popular player is in its own color class. From the theorem, the adaptive algorithm gives \vnar $n_u$, which is a 2-approximation for the optimal \vnar $n_p + n_u - 1$. 
\end{proof}

Interestingly, we show that adaptivity is only necessary when there is uncertainty between popular and unpopular players. 
If we restrict the instance such that uncertain edges exist \emph{only} within the set of popular players and within the set of unpopular players (implying all matches between popular and unpopular players are deterministic), we can achieve the same bound as \Cref{thm:finaladaptive} without adaptivity.

\begin{restatable}{theorem}{chromaticrestricted} \label{thm:chromatic1}
 Suppose that for each uncertain edge in the strength graph, the end points are either both popular or both unpopular.
 There is a seeding with \vnar at least $(n_p+n_u-1)-(k-1)$, where $k$ is the chromatic number of $G_p$.
\end{restatable}
\setcounter{personalcountertwo}{\value{theorem}}
\setcounter{personalcounterthree}{\value{section}}

\subsection{Notation and Preliminaries}
\label{sec:notation-prelim}
Let $n = |N|$, where $N$ denotes the set of players in the input.
The \emph{strength graph} $G = (N,E)$ has a directed edge between every pair of players, and a weight in $[0,1]$ for each edge. 
We assume the match outcomes are independent. 

Recall that a seeding (usually denoted $\sigma$) is a permutation of the players, that determines the order of challengers in a \ctc tournament. The player $\sigma(1)$ is the initial champion, and $\sigma(2)$ is the first challenger. The winner of this match is challenged by player $\sigma(3)$, and so on. Hence,\hide{ in any \ctc tournament,} there are exactly $n-1$ matches. The \emph{value of a match} is determined by the popularity of the winner of the match, and the \emph{value of a seeding} is the sum of the values of the matches played. Note that certain matches may have uncertain outcomes, and hence the value of a seeding --- as well as the sequence of matches played ---may be probabilistic. 

\paragraph{Sub-seedings.} A seeding $\sigma_0$ is said to be a \emph{sub-seeding} of $\sigma$ if $\sigma_0$ is a contiguous sub-permutation of $\sigma$.
It is helpful to describe a sub-seeding using intervals on players.
That is, for two players $a, b \in N$, we use $\sigma[a,b]$ to denote the sub-seeding that includes all players between $a$ and $b$ in the permutation (including $a$ and $b$).
The sub-seedings $\sigma[a,b)$, $\sigma(a,b]$, $\sigma(a,b)$, exclude $b$, $a$, and both $a$ and $b$ respectively.
For a player $a \in N$, we use $\sigma[*,a)$ to denote the set of all players introduced until $a$, excluding $a$.
The sub-seedings $\sigma[a,*]$, $\sigma(a,*]$, and $\sigma[*,a]$ are defined analogously.

\paragraph{Arborescences and Backbones.} A \emph{caterpillar arborescence} (arborescence, in short) is a directed graph where removing all vertices with out-degree zero yields a directed path. 
A seeding $\sigma$ of the players induces several spanning arborescences among the players, depending on the outcomes of the probabilistic games. 
Note that if the seeding results only in deterministic games, then the spanning arborescence is unique. 
It is important to note that \emph{an arborescence implies a seeding as well as a sequence of match outcomes}, some of which may be probabilistic. This observation is used critically in our hardness proof.

For an arborescence, the \emph{backbone} refers to the path obtained after the removal of the leaves.
Note that the players in the backbone are \hide{precisely the ones who} those that win at least one game in the tournament.
Any player who loses a game upon introduction has out-degree zero in the arborescence and is called the \textit{child} of the player it is beaten by.

\noindent{\it Convention about backbone.} 
We note that an edge $u \rightarrow v$ on the backbone\hide{ of an arborescence} indicates the following: player $v$ was the champion when player $u$ was introduced, and $u$ subsequently defeated $v$. 
In all figures, we position player $u$ to the left of player $v$.
\begin{figure}
    \centering    
    \includegraphics[width=\linewidth]{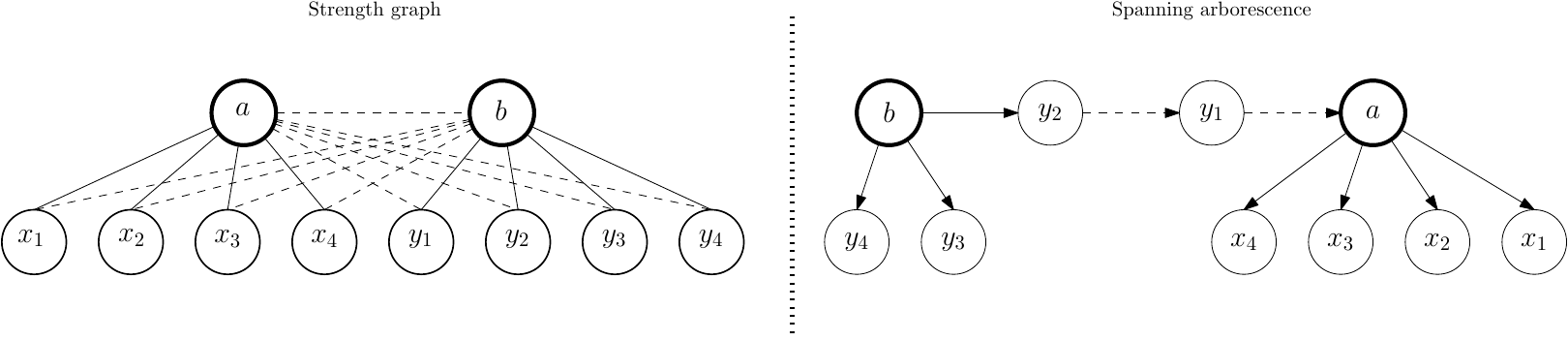}
    \caption{
    Illustration of a strength graph, and a spanning arborescence.
    The instance contains two popular players, $a$ and $b$, and eight unpopular players $x_1, \ldots, x_4$ and $y_1, \ldots, y_4$.
    Solid(dashed) edges are deterministic(uncertain).
    Edges between unpopular players are uncertain and omitted for clarity. The backbone of the arborescence is formed by the vertices $\{b, y_2, y_1, a\}$.
    The \vnar is 7, and is witnessed by the seeding $\langle a, x_1, x_2, x_3, x_4, y_1, y_2, b, y_3, y_4 \rangle$.
    }
    \label{fig:example-with-pu-prob-edges}
\end{figure}

\paragraph{Arborescence with the Minimum Value.} Note that for a given seeding, there can be many resulting arborescences as a result of the sequence of probabilistic outcomes. 
The one that yields the lowest value is of special interest to us. In fact, when we refer to the \vnar score of a seeding, we are actually referring to the value obtained by the least value arborescence. Thus, \vnar of the tournament $\tau$ is the maximum \vnar score among all possible seedings.

Mathematically, let $\sigma$ be a seeding and let $\mathcal{A}(\sigma)$ denote the set of all arborescences that can result from $\sigma$ under different realizations of the probabilistic games. 
We use $\operatorname{score}(A)$ to denote the value of an arborescence $A$; that is, the sum of the values of the matches played.
Moreover, the score of a sub-seeding $\sigma_0$ of $\sigma$ in the arborescence/outcome $A_\sigma$ is the sum of the popularity values of the matches played by players in $\sigma_0$ in the outcome $A_\sigma$.
The $\vnar$ score of $\sigma$ is defined as
$\vnar(\sigma) \;=\; \min_{A \in \mathcal{A}(\sigma)} \operatorname{score}(A).$ The $\vnar$ of the tournament is then defined as $\vnar (\tau)\;=\; \max_{\sigma \in \Sigma} \vnar(\sigma)$, where $\Sigma$ is the set of all possible seedings.

The following result about the {\it minimum-value arborescence} is a critical tool in our analysis.

\begin{lemma}
\label{lem:min-val-arborescence}
\label{backbone}
For any given seeding $\sigma$, a minimum-value arborescence can be computed in time $\bigoh(n^2)$. 
\end{lemma}

\begin{proof}
    Let $N = [n]$ be the set of $n$ players.
    For each $i \in [n]$, we compute 
    \begin{itemize}[itemsep=0pt]
        \item the set $W[i]$ that is the set of possible winners after the introduction of the $i^{th}$ player in $\sigma$;
        \item a $|W[i]|$-dimensional vector $\vect[i]$ that contains in each dimension the minimum possible score of the tournament when the corresponding player in $W[i]$ is the champion after the introduction of the $i^{th}$ player in $\sigma$.
        
        For a player $a \in W[i]$, we use $\vect[i](a)$ to denote the entry in the dimension corresponding to $a$ in that vector.
    \end{itemize}
    Note that by definition, the minimum value in $|\vect[n]|$ is the \vnar of the seeding $\sigma$.
    The values can be computed as follows.
    Let $v_i$ denote the $i^{th}$ player introduced in the tournament.
    When $i=1$, we have $W[1] = \{v_1\}$ and $\vect[1](v_1) = 0$.
    When $i>1$, we do the following.
    Let $S_{i-1} \subseteq W[i-1]$ denote the set of players in $W[i-1]$ that are deterministically beaten by $v_i$; that is, $S_{i-1} = \{ q \in W[i-1] \mid (v_i, q) \in \mathcal{T}\}$.
    Let $T_{i-1} \subseteq W[i-1]$ denote the set of players in $W[i-1]$ that deterministically beat $v_i$; that is, $T_{i-1} = \{ q \in W[i-1] \mid (q, v_i) \in \mathcal{T}\}$.
    If $T_{i-1}=W[i-1]$, then we set $W[i] = W[i-1]$, and $\vect[i](v_j) = \vect[i-1](v_j)+\pop(v_j)$ for each $v_j \in W[i-1]$.
    Note that in this case, all players in $W[i-1]$ deterministically beat $v_i$, and thus form the set of possible winners.
    Else we are in the case where $T_{i-1} \subsetneq W[i-1]$, and we set $W[i] = W[i-1] \setminus S_{i-1} \cup \{v_i\}$.
    This is because the players in $S_{i-1}$ are deterministically beaten by $v_i$ and there is at least one player in $W[i-1]$ that does not deterministically beat $v_i$.
    We set $\vect[i](v_i) = \min_{v_j \in W[i-1] \setminus S[i-1]} \vect[i](v_j)+\pop(v_i)$; and for each $v_j \in W[i-1] \setminus S_{i-1}$, we set $\vect[i](v_j) = \vect[i-1](v_j)+\pop(v_j)$.

    The \vnar of this seeding is given by $\min(\vect[n])$, and an arborescence that witnesses this can be obtained by backtracking.
\end{proof}

\section{The Hardness of Maximizing \vnar}
\label{sec:hardness-vnar}

In the \hampath problem, we are given a directed graph and the goal is to decide if there is a path that visits each vertex exactly once (such a path is called a \textit{Hamiltonian path}).
This problem has been shown to be \nph even in simple digraphs ~\cite{DBLP:books/daglib/0022205}, where any two vertices have at most one edge between them\footnote{We note that Theorem 6.1.3 in \cite{DBLP:books/daglib/0022205} shows that the closely related \textsc{Hamiltonian Cycle} is hard, and a simple modification (removal of certain edges) gives us the desired result.}. We prove the hardness of maximizing \vnar via a reduction from \hampath.

Let $D=(V,\hat{E})$ be an instance of \hampath where $D$ is a simple directed graph on $\hat{n}$ vertices.
We construct an instance $\tourn$ with strength graph $G=(V_p \cup V_u, E)$ as follows. 
\begin{itemize}[itemsep=0pt]
    \item $V_p = \{ v_p \colon v \in V \}$ and all players in $V_p$ have popularity $1$. They form the set of popular players.
    Note that the number of popular players $n_p = \hat{n}$.
    For the sake of simplicity, we sometimes refer to popular players and the popular vertices interchangeably;
    \item $V_u = \{ v_u \colon v \in V \}$ and all players in $V_u$ have popularity $0$. They form the set of unpopular players. The player $v_u$ is called the \textit{unpopular copy} of $v_p$, and $v_p$ is called the \textit{popular copy} of $v_u$;
    \item For each $v \in V$, add $v_p \beats v_u$ to $E$ (that is, $v_p$ beats $v_u$);
    \item For each pair $\{a, b\} \subseteq V$, add $a_u \beats b_p$ to $E$;
    \item For each pair $\{a, b\} \subseteq V$, 
    \begin{enumerate}[itemsep=0pt,wide=10pt]
        \item if $a \beats b \in \hat{E}$, we add $a_p \beats b_p$ to $E$,
        \item else (there is no edge between $a$ and $b$ in $\hat{E}$, and) we add an uncertain edge between $a_p$ and $b_p$ to $E$;
    \end{enumerate}
    \item For any pair $\{a, b\} \subseteq V_p \cup V_u$ that has not been considered before, we add an uncertain edge between $a$ and $b$ to $E$.
\end{itemize}

The following lemma establishes the equivalence between \hampath and maximizing \vnar. 

\begin{lemma}
    \label{lem:vnar-nph-correctness}
    $D$ has a Hamiltonian path if and only if there is a seeding in $\tourn$ with $\vnar \ge 2n_p-1$.
\end{lemma}
\begin{proof}
For the forward direction, let $v^1 \beats \cdots \beats v^{n_p}$ be a Hamiltonian path in $D$. 
Consider a seeding corresponding to the following backbone.
\begin{figure}[H]
    \centering
    \includegraphics[width=0.5\linewidth]{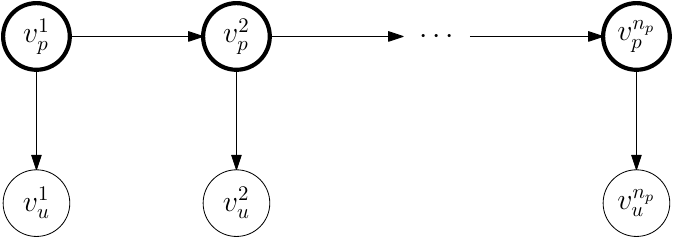}
    \caption{The backbone witnessed by the seeding $\langle v^{n_p}_p, v^{n_p}_u, \dots, v^1_p, v^1_u \rangle$}
    \label{fig:hp-tfp-backbone}
\end{figure}\noindent 
All the games are deterministic and thus the seeding has \vnar $2n_p-1$.

For the reverse direction, let $\sigma$ be a seeding with $\vnar \ge 2n_p-1$.
Consider the arborescence, denoted by $A_\sigma$, which gives the score $2n_p-1$. We note that among all the arborescences that are produced by $\sigma$, $A_\sigma$ has the smallest score. 

Since there are exactly $2n_p-1$ games, it follows that every game must have score 1. That is, it must be won by a popular player. This implies that all the popular players are in the backbone and every unpopular player must lose to its popular copy. This implies that every unpopular player is a leaf of its popular copy. Refer to \Cref{fig:hp-tfp-backbone} for an illustration. 

Next we will argue that the backbone must actually be a deterministic path for the \vnar to be at least $2n_p-1$. Consider the first two consecutive popular players, denoted by $v_p$ and $x_p$, in the backbone such that they share an uncertain edge and $x_p$ is introduced before $v_p$ in the seeding. Due to the assumption on the edge $(v_p, x_p)$, it follows that the backbone consists of a directed path from $x_p$ to the first player in the backbone. Additionally, due to the aforementioned property of the unpopular players being leaves of backbone vertices, it follows that in the seeding $\sigma$, player $v_p$ is followed by its unpopular copy $v_u$. 

Consider the probabilistic event in which $x_p$ beats $v_p$ and then is beaten by the unpopular player $v_u$ who then\hide{ is introduced next in the seeding. Following that $v_u$} wins all the remaining games. This results in an arborescence $A'_\sigma$ which is identical to $A_\sigma$ up to the introduction of $v_p$ but beyond that game gives value 0. Note that $A'_\sigma$ has score at least 1 less than $A_\sigma$. This is due to the game between $v_u$
and $x_p$ which gives value 0. But this is a contradiction to the assumption of $A_\sigma$. Therefore, we can conclude that the backbone cannot have a single uncertain edge and hence it must be a deterministic path on the $n_p$ popular players, implying that it is a Hamiltonian path in $D$. 
\end{proof}

In the following sections, we show that the above reduction establishes the hardness of finding even approximately good solutions.
Our hardness results rely on the following result by Bj{\"o}rklund, Husfeldt, and Khanna \cite[Theorem 4]{DBLP:conf/icalp/BjorklundHK04}, on finding long paths.

\begin{proposition}[\cite{DBLP:conf/icalp/BjorklundHK04}]
\label{prop:source-hardness}
    Unless \textup{\textsf{P=NP}}, there is no deterministic, polynomial-time algorithm that finds a path of length $(\hat{n}-1)/\hat{n}^{1-\varepsilon}$ for any fixed $0 < \varepsilon \le 1$ even in simple\footnote{Though not explicitly stated so, one can observe that the given reductions involve only simple directed graphs.} directed graphs on $\hat{n}$ vertices where a Hamiltonian path is guaranteed to exist.
\end{proposition}

\subsection{Additive Hardness of Approximation}
\label{subsec:additive-hardness-approx}
Let $D=(V,\hat{E})$ be a simple directed graph on $\hat{n}$ vertices where a Hamiltonian path is guaranteed to exist.
We construct an instance $\tourn$ with strength graph $G=(V_p \cup V_u, E)$ on $n = \hat{n}+\hat{n}$ players following the same reduction that establishes \Cref{lem:vnar-nph-correctness}. 
Note that $n_p = n_u = \hat{n}$. 
Observe that in the reduced instance $\tourn$, it is straightforward to find a seeding with \vnar $n_p-1$ by first introducing the popular players in an arbitrary order, followed by the unpopular players in an arbitrary order.
Towards establishing the hardness of finding an approximately good seeding, we study the interesting case where the computed seeding has a score strictly greater than $n_{p}-1$.
Suppose that a polynomial-time algorithm computes a seeding $\sigma$ with \vnar $n_p-1+\gamma$ for some $\gamma > 0$.

The following lemmas on the structure of a minimum-value arborescence $A_\sigma$ (obtained by applying \Cref{lem:min-val-arborescence}) allow us to establish the inapproximability.
\begin{lemma}
    \label{clm:backbone-pop-lb}
    \label{clm:arborescence-pop-lb}
    There are at least $\gamma-1$ popular players on $A_\sigma$ that have their unpopular copies as a child. 
\end{lemma}
\begin{proof}
    Assume for contradiction that there are $q $ many players satisfying the property in the lemma statement such that $q< \gamma-1$. We denote the set of these (popular) players by $C_1$ and the rest of the popular players by $C_2$.
    The set of all games in $A_\sigma$ can be partitioned into the following types
    \begin{enumerate}[label=($t_\arabic*$).]
        \item a player in $C_1$ beats a player in $C_1$;
        \item a player in $C_1$ beats a player in $C_2$;
        \item a player in $C_2$ beats a player in $C_1$;
        \item a player in $C_2$ beats a player in $C_2$;
        \item a player $v_p\in C_1$ beats its unpopular copy $v_u$;
        \item a player in $v_p \in C_2$ beats its unpopular copy $v_u$;
        \item an unpopular player $v_u$ beats a player $v'_p \in C_1$;
        \item an unpopular player beats a player in $C_2$;
        \item an unpopular player beats an unpopular player.
    \end{enumerate}
    For the sake of convenience, let $t_i$ also denote the number of games of type $i$.
    Let $s_i$ denote the total value of the games in $\sigma$ of type $t_i$. 
    By construction, note that $t_i = s_i$, for all $i\in [6]$ and $s_i=0$ for $7\leq i\leq 9$. Hence, 
    the total value of $A_\sigma$ is $s_1+\dots+s_6$. 
    Observe that $s_6 = t_6 \le t_7 + t_8 + t_9$ since in any game of type $t_6$ where a player $v_p \in C_2$ beats its unpopular copy $v_u$, it must be the case that $v_u$ wins some game and is not a leaf in the arborescence (otherwise, the popular player $v_p \in C_1$ and not $C_2$, by definition). 

    Moreover, using the definitions of $C_1$ and $C_2$, we derive the following inequalities. 
    In every game of the type $t_1$, $t_3$ and $t_7$, a player in $C_1$ loses. 
    Thus, it follows that $t_1+t_3+t_7 \le |C_1|$. 
    Since $s_i=t_i$ for each $i\in [6]$, and $q=|C_1|$, we have $s_1 + s_3 + t_7 \le q$. 
    Similarly, arguing for games involving players in $C_2$, we obtain $s_2+s_4+t_8 \le |C_2|$. Combining these two yields
    \[
    s_1 + s_2 + s_3 + s_4+ t_7 + t_8 \le  |C_1| + |C_2|. 
    \]
    Moreover, we also have $s_5 = q$, and hence 
    \[ 
    s_1 + \ldots + s_5 + t_7 + t_8 \le 
    |C_1| + |C_2| +s_5 = (q) + (n_p-q) + q = n_p+q
    \tag{I}\label{eq:partial-scores}  
    \]
    
    To achieve our desired result, we argue by ``accounting''. 
    Recall that $A_\sigma$ is the arborescence corresponding to some probabilistic event that has the least score for the seeding $\sigma$.
    We account any game of type $t_6$ to the unpopular player who loses it. 
    Consider a game $g$ of type $t_6$, played between $v_p$ and $v_u$.  
    Note that by the aforementioned property of $C_2$, the unpopular player $v_u$ is the champion right before it loses the game to its popular copy $v_p$. 
    
    Case I: Consider the last popular player (if it exists), denoted by $w_p$, introduced before the game $g$. 
    We note that $w_p$ must be losing to some unpopular player in a game $g'$ (of type $t_7$ or $t_8$), who is possibly distinct from $v_u$. 
    Depending on whether $w_p$ is in $C_1$ or $C_2$, we account the game $g$ to the game $g'$ (that $w_p$ loses). 
    For each game of type $t_6$, if there exists such a player $w_p$, then we can map it to a game of type $t_7$ or $t_8$. Consequently, the left-hand-side of \cref{eq:partial-scores} gives the total score of the tournament. 
    But that yields the contradiction that the score is at most $n_p+q < n_p-1 +\gamma $.

    Case II: We are left to argue on the case where no popular player (such as $w_p$ above) is introduced before a game of type $t_6$, say between $v_p$ and $v_u$. 
    Such a game is so far {\it unaccounted} for. 
    Observe that there could be only one such unaccounted game in the entire tournament, because if there are two or more $t_6$ games, then the popular player in the first $t_6$ game is introduced before the popular player of the second $t_6$ game. 
    Hence, the second (and all subsequent) game(s) are accounted for. 
    Suppose that the game $g$ is unaccounted for. 
    So, all the players introduced before $v_u$ are unpopular.

    The idea for this case is to construct a new seeding that is easier to analyze but has the same score. 
    Let $S_u$ be the set of (unpopular) players other than $v_u$ introduced before the game $g$. 
    Observe that $S_u$ is non-empty, otherwise $v_u$ is the first player introduced and the game $g$, where $v_p$ beats $v_u$ would be a game of type $t_5$ and not $t_6$, thus contradicting the definition of Case II. 
    Moreover, all the games that players in $S_u$ participate in involve only unpopular players.
    Consider the new seeding $\sigma'$ where the players in $S_u$ are introduced at the very end of the tournament and $v_u$ becomes the very first player introduced in the tournament (and $v_p$ is the second player).
    Since the score of $\sigma'$ is same as that of $\sigma$, and there is no unaccounted game of type $t_6$ (the game between $v_p$ and $v_u$ in $\sigma'$ is of type $t_5$), we have from the analysis of Case I that the score of $\sigma$ is at most $n_p-1+\gamma$, a contradiction.
\end{proof}

\begin{lemma}
    \label{clm:backbone-unpop-ub}
    \label{lem:backbone-unpop-ub}
    There are at most $n_p-\gamma$ unpopular players on the backbone of $A_\sigma$.
\end{lemma}
\begin{proof}
    Assume for contradiction that there are $q>n_p-\gamma$ players satisfying the property in the claim statement.
    Due to the games won by the $q$ players satisfying the claim statement, the value of $\sigma$ is at most $2n_p-1-q < 2n_p-1-(n_p-\gamma) = n_p-1+\gamma$, a contradiction.
\end{proof}

\begin{lemma}
    \label{clm:backbone-pop-det-path}
    \label{lem:backbone-pop-det-path}
    Suppose that $a_p$, $b_p$ are two popular players such that in $A_\sigma$, $a_p$ is introduced after $b_p$, and both $a_u$ is a child of $a_p$ and $b_u$ is a child of $b_p$.
    Suppose that only popular players are introduced in $A_\sigma(a_u, b_u)$, and let $X$ denote that set.
    Then, there is a deterministic path $P$ from $a$ to $b$ in $D$ with intermediate vertices corresponding to those from $X$ alone, that is, $P \subseteq X \cup \{a,b\}$.
    Moreover, such a path can be computed in polynomial time.
\end{lemma}
\begin{proof}
    Suppose that there is no such path from $a$ to $b$ in $D$.
    Let $v_p \in X \cup \{b_p\}$ be the popular player who is introduced the latest in $A_\sigma[b_p,a_u)$ such that there is a deterministic path from $v$ to $b$ with intermediate vertices (if at all any) from $X$.
    Note that such a $v_p$ exists since it could potentially be $b_p$ itself.
    Note that $v_p \ne a_p$ by our supposition.
    We define a new arborescence (event) $A'_{\sigma}$ depending on who is the champion, say $c_p$, right after the introduction of $v_p$ in $A_\sigma$.
    \begin{itemize}
        \item If $c_p = v_p$ or $c_p$ deterministically beats $v_p$, then $A'_\sigma$ is the event in which $c_p$ beats all players introduced in $\sigma(v_p, a_u)$, loses to $a_u$, and then $a_u$ beats all players in $\sigma(a_u,*]$.
        Note that all players introduced in $\sigma(v_p, a_u)$ do not deterministically beat $v_p$ (by definition of $v_p$).
        \item Else $c_p$ has an uncertain edge with $v_p$ and $A'_\sigma$ is the event in which $v_p$ beats $c_p$, beats all players introduced in $\sigma(v_p, a_u)$, loses to $a_u$, and then $a_u$ beats all players in $\sigma(a_u, *]$.
        Similar to the previous case, note that all players in introduced in $\sigma(v_p, a_u)$ do not deterministically beat $v_p$ (by definition of $v_p$).
    \end{itemize}
    
    In all the above cases, the event $A'_\sigma$ has score that is at least one less than that of $A_\sigma$ (due to the unpopular player $s_u$ winning), contradicting the definition of $A_\sigma$.
    
    We note that such a path can be computed in polynomial time as follows: find a path from $a$ to $b$ by applying breadth-first search starting at vertex $a$ in the graph $D$ restricted to the vertices in $X \cup \{a,b\}$.
\end{proof}

Now, we are ready to prove our result.
Applying the pigeonhole principle on \Cref{clm:backbone-pop-lb,clm:backbone-unpop-ub}, we obtain a sequence $S$ of $(\gamma - 1) / (n_p-\gamma)$ popular players such that there is no unpopular player on the backbone between any two players in $S$.
After repeated application of \Cref{clm:backbone-pop-det-path} (on adjacent players in $S$), we obtain a deterministic path of length $(\gamma-1)/(n_p-\gamma)-1$ in the graph $D$.
Setting $\gamma=(1+n_p\delta)/(1+\delta)$ where $\delta=1+(n_p-1)/n_p^{1-\varepsilon}$, we obtain a deterministic path of length $(n_p-1)/n_p^{1-\varepsilon}$ in $D$, that is \nph to find.
\shortv{
    UNCOMMENT THE FOLLOWING TO SEE WHY THIS CHOICE OF GAMMA WORKS
    \begin{align*}
        (\gamma - 1) / (n_p-\gamma) -1 &= \frac{\frac{1+n_p\delta}{1+\delta}-1}{n_p-\frac{1+n_p\delta}{1+\delta}} -1 \\
        &= \frac{n_p\delta-\delta}{n_p-1}-1  = \delta-1\\
        &= \frac{n_p-1}{n_p^{1-\varepsilon}}.
    \end{align*}
}
So unless $\textsf{P=NP}$, there is no polynomial-time algorithm that can find a seeding with score at least $n_p-1+\gamma = (2n_p-1)-(n_p-\gamma)$.

Recall that $D$ is assumed to contain a Hamiltonian path.
By \Cref{lem:vnar-nph-correctness}, we know that the reduced instance has $\vnar$ $2n_p-1$.
Thus, we rule out a polynomial-time additive approximation algorithm of factor at most $n_p-\gamma$. Moreover, we show the following.

\begin{lemma}
\label{clm:calcution}
For all $\alpha >0$ and $0<\varepsilon <1$ and sufficiently large $n_p\in \mathbb{N}$, we have $n_p^{1-\varepsilon}/(1+\alpha) \le n_p-\gamma$.
\end{lemma}

\begin{proof}
Replacing the value of $\gamma$ and $\delta$ from earlier, we get

\begin{align*}
    n_p-\gamma & =   n_p - \frac{1+n_p\delta}{1+\delta} =   n_p-\frac{1+n_p\left(\frac{n_p-1}{n_p^{1-\varepsilon}}+1\right)}{1+\left(\frac{n_p-1}{n_p^{1-\varepsilon}}+1\right)} \\
        & = n_p - \frac{n_p^{1-\varepsilon}+n_p(n_p-1)+n_p \cdot n_p^{1-\varepsilon}}{2n_p^{1-\varepsilon}+n_p-1}\\
        & = \frac{2n_p\cdot n_p^{1-\varepsilon}+n_p(n_p-1)-n_p^{1-\varepsilon}-n_p(n_p-1)-n_p\cdot n_p^{1-\varepsilon}}{2n_p^{1-\varepsilon}+n_p-1}\\   
        & = \frac{n_p\cdot n_p^{1-\varepsilon}-n_p^{1-\varepsilon}}{2n_p^{1-\varepsilon}+n_p-1}\\
        & = n_p^{1-\varepsilon} \cdot \frac{n_p - 1}{n_p -1 + 2n_p^{1-\varepsilon}}  = n_p^{1-\varepsilon} \cdot \frac{1}{1 + \frac{2n_p^{1-\varepsilon}}{n_p-1}} \, .\\
\end{align*}

\noindent Choosing $n_p$ so that $\alpha \ge \frac{2n_p^{1-\varepsilon}}{n_p-1}$, or $n_p \ge (4/\alpha)^{1/\varepsilon}$, we obtain the lemma.
\end{proof}

Overall, we have the following hardness of approximation.

\thmadditivepnp*

\begin{sloppypar}
Bj{\"o}rklund, Husfeldt, and Khanna \cite{DBLP:conf/icalp/BjorklundHK04} also showed that assuming the \textsf{ETH} \cite{DBLP:journals/jcss/ImpagliazzoP01} (a stronger conjecture than $\textsf{P}\ne\textsf{NP}$), no polynomial-time algorithm can find a path of length $f(n_p) \log^2 n_p$ for any polynomial-time computable and non-decreasing function $f \in \omega(1)$ even in simple directed graphs on $n_p$ vertices where a Hamiltonian path is guaranteed to exist.
Fix any such function $f$.
Observe that $f/2 \in \omega(1)$.
Setting $\gamma = (1+n_p\delta)/(1+\delta)$ where $\delta=1+(f(n_p)/2)\log^2 n_p$, we obtain a deterministic path of length $(f(n_p)/2) \log^2 n_p$ in $D$, that is hard to find.
Following the same arguments as before, we rule out approximation factors that are at most $n_p-\gamma = (n_p-1)/(2+(f(n_p)/2)\log^2 n_p)$.
In particular, we rule out factor $n_p/f(n_p)\log^2 n_p \le (n_p-1)/(2 \times (f(n_p)/2)\log^2n_p) \le n_p-\gamma$.
Observe that when $n_p^\varepsilon > (f(n_p)/2) \log^2 n_p$ the approximation factor ruled out by assuming \textsf{ETH} is higher than the approximation factor ruled out by assuming $\textsf{P}\ne\textsf{NP}$.
Thus, we get a stronger inapproximability result depending on the choice of $f$.
Overall, we have the following.
\end{sloppypar}

\thmc*

\subsection{Multiplicative Hardness of Approximation}

In this subsection, we establish a different hardness of approximation, via a slightly different reduction compared to the one in \Cref{subsec:additive-hardness-approx}.
Let $D=(V,\hat{E})$ be a simple directed graph on $\hat{n}$ vertices where a Hamiltonian path is guaranteed to exist.
Let $d = \hat{n}^2$.
We construct an instance $\tourn$ with strength graph $G=(V_p \cup V_u, E)$ on $n = \hat{n}+\hat{n}d$ players as follows.
\begin{itemize}[itemsep=0pt]
    \item $V_p = \{ v_p \colon v \in V \}$ and all players in $V_p$ have popularity $1$. They form the set of popular players; note that $n_p = \hat{n}$);
    \item $V_u = \{ v_{u,i} \colon ~v \in V, i \in [d] \}$ and all players in $V_u$ have popularity $0$. They form the set of unpopular players. For each $v \in V$, we call the players in $\{v_{u,i} \colon i \in [d]\}$ as the unpopular copies of $v$;
    \item For each $v \in V$ and $i \in [d]$, add $v_p \beats v_{u,i}$ to $E$ (that is, $v_p$ beats its unpopular copies);
    \item For each pair $\{a,b\} \subseteq V$, 
    \begin{enumerate}[itemsep=0pt,wide=10pt]
        \item if $a \beats b \in \hat{E}$, we add $a_p \beats b_p$ to $E$;
        \item else (there is no edge between $a$ and $b$ in $\hat{E}$, and) we add an uncertain edge between $a_p$ and $b_p$ to $E$.
    \end{enumerate}
    \item For any pair $\{a,b\} \subseteq V_p \cup V_u$ that has not been considered before, we add an uncertain edge between $a$ and $b$ to $E$.
\end{itemize}

Suppose that a polynomial-time algorithm computes a seeding $\sigma$ with \vnar at least $\ell d + n_p - 1$ for some $\ell = (n_p)^\varepsilon$ where $\varepsilon>0.5$.
Note that the \vnar of a seeding can be computed in polynomial time, \Cref{lem:min-val-arborescence}.
Let $A_{\sigma}$ denote the spanning arborescence of $\sigma$ that corresponds to the (sequence of) following probabilistic events:
\begin{itemize}[itemsep=0pt]
    \item whenever there is a probabilistic game between an unpopular and a popular player, the unpopular player wins.
\end{itemize}
By construction, we have the following equivalent description of the probabilistic events that generate the arborescence $A_\sigma$: a game between a popular player $v_p$ and an unpopular player $w_u$ in $A_\sigma$ is won by $v_p$ implies that $v=w$ (that is, the latter is a copy of the former), else it is won by $w_u$.
We begin with the following property. 

\begin{observation}\label{obs:children-are-copies}
    \label{lem:children-are-copies}
    Any popular player in the backbone of $A_\sigma$ who has an unpopular child must be the latter's copy.
\end{observation}
Due to the definition of $\sigma$, it follows that it has a prefix whose \vnar is exactly $\ell d + n_p-1$. Our argument will focus on this prefix. Formally, we proceed as follows.

Let $\sigma_0$ be the minimal-prefix sub-seeding of $\sigma$, such that the partial arborescence $A_{\sigma_0}$ obtained from $A_{\sigma}$ by only considering the players in $\sigma_0$ has score exactly $\ell d + n_p-1$.
In what follows, we will establish a set of structural properties about the sub-seeding $\sigma_0$ and the corresponding arborescence $A_{\sigma_0}$. 
These properties will enable us to obtain a sufficiently long path in $D$ which establishes the inapproximability of \vnar.
For the ease of exposition and to aid modularity, we break the analysis into paragraphs.

\paragraph{Decomposing the seeding into blocks.}
We will begin by decomposing the seeding into blocks punctuated by popular players that satisfy certain specific properties. The underlying idea is that these blocks allow us to identify deterministic paths, which are either present in the backbone or must be inferred from the uncertain edges. In certain scenarios, there is a large score increase if certain probabilistic outcomes are different. This score increase, in turn, allows us to identify a deterministic path such that the score increase is proportional to the length of that path in $D$.

We begin our formal discussion by identifying a scenario where a deterministic path in $D$ can be easily identified. 
Let $v^1_u$ be the last-introduced (in $A_{\sigma_0}$) unpopular child of $v^1_p$ and $v^2_u$ be the first-introduced (in $A_{\sigma_0}$) unpopular child of $v^2_p$.

\begin{definition}\label{def:clean-pair}
   We define a {\em clean pair} to be a pair of popular players $(v^1_p, v^2_p)$ if they satisfy the following properties in $A_{\sigma_0}$:
    \begin{enumerate}[label=(\roman*)]
        \item\label{prop:block1} both are on the backbone of $A_{\sigma_0}$ with $v^2_p$ introduced after $v^1_p$; 
        \item\label{prop:one-U-child} each of them has at least one unpopular child;
        \item\label{prop:non-interleaving} there is no other popular player on the backbone between them who has an unpopular child; and
        \item\label{item:clean-block-all-pop} either the backbone between them contains only popular players, or there is a deterministic path from $v^2_p$ to $v^1_p$ in $D$ such that the intermediate vertices correspond to the popular players introduced between $v^1_u$ and $v^2_u$ in $\sigma_0$. 
    \end{enumerate}
    Moreover, the sub-seeding $\sigma_0[v^1_p, v^2_u]$ is called a {\it clean block}.
\end{definition}

The following lemma identifies a scenario when a deterministic path is easy to detect: the existence of a clean block. 
\begin{lemma}
\label{lem:clean-block}
Let $(v^{i,1}_p, v^{i,2}_p)$, denote a clean block. 
Let $X$ denote the set of popular players introduced in $\sigma_0(v^{i,1}_u, v^{i,2}_u)$. 
Then, there exists a deterministic path $P$ from $v^{i,2}$ to $v^{i,1}$ in $D$ such that the intermediate vertices correspond to those from $X$ alone, that is, $P \subseteq X \cup \{v^{i,1}_p, v^{i,2}_p\}$.
Moreover, such a path can be computed in polynomial time.
\end{lemma}

\begin{proof} 
    By the definition of a clean block, it is sufficient to consider the case where in $A_{\sigma_0}$, only popular players are introduced between $v^{i,1}_u$ and $v^{i_2}_u$ (the first case in Condition \ref{item:clean-block-all-pop}).
    If we exhibit a new probabilistic event within the sub-seeding $\sigma_0$ in which $v^{i,2}_u$ wins all its games after introduction and there is a decrease in score, then we also obtain such a probabilistic event in $\sigma$ in which $v^{i,2}_u$ wins all its games after introduction (with the same score since $v^{i,2}_u$ is unpopular).
    Since the new event has score strictly less than $\ell d + n_p - 1$ (due to the decrease), we have a contradiction to the fact that $\sigma$ has \vnar at least $\ell d + n_p -1$.
    Thus, by a proof that closely follows in the lines of the one given for \Cref{clm:backbone-pop-det-path}, we have our desired result.
    We add it below for the sake of completeness.

    Suppose that there is no such path from $v^{i,2}$ to $v^{i,1}$ in $D$.
    Let $x_p \in X \cup \{v^{i,1}_p\}$ be the popular player who is introduced the latest in $A_\sigma[v^{i,1}_p,v^{i,2}_u)$ such that there is a deterministic path from $x$ to $v^{i,1}$ with intermediate vertices (if at all any) from $X$.
    Note that such a $x_p$ exists since it could potentially be $v^{i,1}_p$ itself.
    Note that $x_p \ne v^{i,2}_p$ by our supposition.
    We define a new arborescence (event)e $A'_{\sigma_0}$ depending on who is the champion, say $c_p$, right after the introduction of $x_p$ in $A_{\sigma_0}$.
    \begin{itemize}
        \item If $c_p = x_p$ or $c_p$ deterministically beats $x_p$, then $A'_{\sigma_0}$ is the arborescence in which $c_p$ beats all players introduced in $\sigma_0(x_p, v^{i,2}_u)$, loses to $v^{i,2}_u$, and then $v^{i,2}_u$ beats all players in $\sigma_0(v^{i,2}_u,*]$.
        Note that all players introduced in $\sigma_0(x_p, v^{i,2}_u)$ do not deterministically beat $x_p$ (by definition of $x_p$).
        \item Else $c_p$ has an uncertain edge with $x_p$ and $A'_{\sigma_0}$ is the arborescence in which $x_p$ beats $c_p$, beats all players introduced in $\sigma_0(x_p, v^{i,2}_u)$, loses to $v^{i,2}_u$, and then $v^{i,2}_u$ beats all players in $\sigma_0(v^{i,2}_u, *]$.
        Similar to the previous case, note that all players in introduced in $\sigma_0(x_p, v^{i,2}_u)$ do not deterministically beat $x_p$ (by definition of $x_p$).
    \end{itemize}
    In all the above cases, the arborescence $A'_{\sigma_0}$ has score that is at least one less than that of $A_{\sigma_0}$ (due to the unpopular player $s_u$ winning).
    By extending $A'_{\sigma_0}$, we also obtain an arborescence in the whole seeding $\sigma$ in which $v^{i,2}_u$ wins all its games after introduction (with the same score as that of $A'_{\sigma_0}$ since $v^{i,2}_u$ is unpopular).
    Since the new arborescence has score strictly less than $\ell d + n_p - 1$ (due to the decrease), we have a contradiction to the fact that $\sigma$ has \vnar at least $\ell d + n_p -1$.

    We note that such a path can be computed in polynomial time as follows: find a path from $v^{i,2}$ to $v^{i,1}$ by applying breadth-first search starting at vertex $v^{i,2}$ in the graph $D$ restricted to the vertices in $X \cup \{v^{i,2},v^{i,1}\}$.
\end{proof}

We will show that $A_{\sigma_0}$ consists of several consecutive clean blocks which allow us to identify a sufficiently long path. 
Towards that, we analyze the possibility that between two popular players in the backbone, whose respective children are their unpopular copies, there exists an unpopular player in the backbone.

\begin{definition}\label{def:tagged-pair}
We define a \emph{tagged pair} as a pair of popular players $(v^1_p, v^2_p)$, illustrated in \Cref{fig:backbone-bad-case}, if they 
satisfy the following properties in $A_{\sigma_0}$:
\begin{enumerate}[label=(\roman*)]
    \item both are on the backbone of $A_{\sigma_0}$ with $v^2_p$ introduced after $v^1_p$;
    \item each of them has at least one unpopular child;
    \item\label{prop:non-interleaving-repeated} there is no other popular player on the backbone between them who has an unpopular child; and
    \item the backbone between them contains at least one unpopular player, and there is no path in the graph $D$ from $v^2_p$ to $v^1_p$ where the intermediate path vertices correspond to the popular players introduced between $v^1_u$ and $v^2_u$ in $\sigma_0$.
\end{enumerate}
Moreover, the members of a tagged pair are called the \emph{(tagged) partners},
and the sub-seeding $\sigma_0[v^{1}_p, v^{2}_u]$ is called a \emph{tagged block}. 
\end{definition}

\begin{figure}
    \centering
    \includegraphics[width=0.9\linewidth]{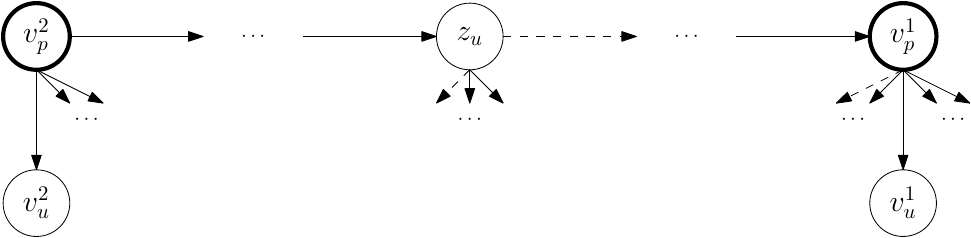}
    \caption{Illustration of a tagged block due to partners $(v^1_p, v^2_p)$. Here, the bold circles represent the popular backbone players, and the normal circles are unpopular players. The solid(dashed) arrows represent the deterministic (uncertain) edges between players.} 
    \label{fig:backbone-bad-case}
    \label{fig:tagged-pair-example}
\end{figure}

Let $\{(v^{i,1}_p, v^{i,2}_p)\}_{i\geq 1}$ denote the set of tagged pairs in $\sigma_0$, indexed in the order in which they are introduced in $\sigma_0$.
Note that condition~\ref{prop:non-interleaving-repeated} ensures that these pairs cannot interleave.
Additionally, we will call the members of a tagged pair to be (tagged) partners.

Consider a tagged pair $(v^{i,1}_p, v^{i,2}_p)$: from now on, we call it the $i^{th}$ pair.
As in the case of clean blocks, we will use $v^1_u$ to denote the last-introduced (in $A_{\sigma_0}$) unpopular child of $v^1_p$ and $v^2_u$ to denote the first-introduced (in $A_{\sigma_0}$) unpopular child of $v^2_p$.

Let $x^i_p$ denote either $v^{i,1}_p$ or the last popular player that is introduced between $v^{i,1}_u$ and $v^{i,2}_u$ such that there is a deterministic path in the graph $D$ from the vertex $x^i$ to $v^{i,1}$ in which the intermediate vertices of the path correspond to the popular players introduced between the tagged partners in $\sigma_0$.
(Note that this deterministic path also exists in the reduced instance within the subgraph induced on the popular players.)

We identify a special player $w^{i}$ depending on who is the champion right after the introduction of $x^i_p$. 
In the following definition, we consider the champion, denoted by $c$, right after $x^{i}_p$ is introduced. Thus, $c$ is either $x^i_p$ or whoever is the champion right before the introduction of $x^i_p$.   
We define $w^i$ as follows. 
\begin{itemize}
\item 
    If $x^i_p$ is {\it not} deterministically beaten by the previous champion, we set $w^i = x^{i}_p$. 
\item
    Else, $w^i = c$. 
\end{itemize}
Note that, irrespective of what value it takes, $w^i$ is a popular player. 
We first prove the following result.

\begin{lemma}
\label{lem:w^i-wins}
Consider any player $y$ introduced in $\sigma_0[x^i_p, v^{i,2}_u)$ such that $y \neq w^i$. 
Then player $y$ either deterministically loses to $w^i$ or has an uncertain edge to $w^i$. 
\end{lemma}
\begin{proof}
    We first argue the case when $y\neq x^i_p$. 
    That is, $y$ is introduced after $x^i_p$. 
    There are two cases depending on: $w^i = x^i_p$ and $w^i \neq x^i_p$. 
    Both cases can be argued similarly, but for clarity, we break them into two parts. 
    
    Suppose that $w^i \neq x^i_p$. 
    Then, $w^i$ is the previous champion before the introduction of $x^i_p$ and this player deterministically beats $x^i_p$. 
    Moreover, if $y$ deterministically beats $w^i$, then $y$ and $w^i$ are both popular players and the former is deterministically beating the latter. 
    Since $y$ is introduced after $x^i_p$, $y$ has a path to $v^{i,1}_p$ in digraph $D$, a contradiction to the definition of $x^i_p$. 
    
    Suppose that $w^i= x^i_p$. 
    If $y$ deterministically beats $w^i$, then the existence of the path from $x^i$ to $v^{i,1}$ in $D$ implies that in fact there is path from $y$ to $v^{i,1}$ in $D$, a contradiction to the definition of $x^i_p$. 
    Note that this case also handles the situation where popular player $x^{i}_p$ is a child of $v^{i,2}_p$ and is thus introduced in $\sigma_0(v^{i,2}_p, v^{i,2}_u)$. 
    
    Hence, as long as $y \neq x^i_p$, $y$ is introduced after $x^{i}_p$. 
    Thus, either $y$ deterministically loses to $w^i$ or has an uncertain edge with $w^i$.
    
    For the special case, where $y=x^i_p$, we note that $y \neq w^i$, implies that $w^i$ is the champion in $\sigma_0$ right before the introduction of $x_i^p$. 
    Additionally, from the definition of $w^i$, it follows that $w^i$ must deterministically beat $x^i_p$. 
    Hence, the lemma is proved.
\end{proof}

We prove the following result, which allows us to analyze an alternate event in which the score of a tagged block may increase. 
\begin{lemma}
    \label{clm:bad-prob-outcome}
    There is a probabilistic event in which $w^i$ beats everyone from its introduction till $v^{i,2}_u$ and then loses to $v^{i,2}_u$.
\end{lemma}
\begin{proof}
    Due to \Cref{lem:w^i-wins}, we know that $w^i$ either deterministically beats or can probabilistically beat each player introduced after it in $\sigma_0[w^i, v^{i,2}_u)$.
\end{proof}

\paragraph{Winning the block.}\label{para:winning-the-block} We refer to the condition of $w^i$ beating everyone till the introduction of $v^{i,2}_u$ and then losing to $v^{i,2}_u$, as $w^i$ {\it winning the block}.

Since $w^i$ is a popular player, $w^i$ winning the block ensures that, in the resulting arborescence, every game after the introduction of $w^i$ up to the introduction of $v^{i,2}$ gives value. Thus, the score up to here is at least as much as the score obtained by $A_{\sigma_0}$. But, there is zero score after the introduction of $v^{i,2}_u$, since the unpopular player $v^{i,2}_u$ wins all the remaining games. 

In what follows, we will analyze the score increase due to $w^i$ winning the block. We will argue that if the score increases, due to $w^i$ winning the block, is upper-bounded, then the score of the seeding up to the start of that block is lower bounded, \Cref{obs:max-score-of-tagged-block}. 
This allows us to identify a prefix of the seeding in which all tagged blocks have large score increases. 

Let $s^i$ denote the change in score from the original arborescence $A_{\sigma_0}$ to the one when $w^i$ is winning the block. If $s^i \leq \ell d/2+ n_p-1$, then we call $i$ a {\it low-scoring} block. 

The following result, \Cref{obs:max-score-of-tagged-block}, allows us to bound the value of the score {\it inside} a tagged block. 
We use the property of $w^i$ winning the block to give an estimate of the score that can be obtained within a tagged block. 
At the heart of this argument is the fact that in $A_{\sigma_0}$ only the popular players on the backbone between the tagged partners contribute to the score by playing each other. 
But since there are unpopular players in the backbone, the score obtained from within a tagged block can actually be higher, such as when $w^i$ wins the block. 
This insight allows us to bound the score of the sub-seeding up to the introduction of the first partner of the first low-scoring (tagged) block,~\Cref{lem:score-first-low-block}.

\begin{lemma}
\label{obs:max-score-of-tagged-block}
The score obtained by the sub-seeding $\sigma_0[v^{i,1}_p, *]$ when $w^i$ wins the block is at most $s^i + n_p-1$. 
\end{lemma}
    \begin{proof}
        Due to the definition of $v^{i,1}_u$, it follows that all other children of the backbone player $v^{i,1}_p$ introduced after $v^{i,1}_u$ must be popular. 
        Similarly, all children of the backbone player $v^{i,2}_p$ introduced before $v^{i,2}_u$ must be popular. 
        Additionally, due to the condition~\ref{prop:non-interleaving-repeated} of tagged pairs no popular player introduced in  $\sigma_0(v^{i,1}_u, v^{i,2}_u)$ can have an unpopular child. 
        Thus, the total score given by $A_{\sigma_0}$ in $\sigma_0(v^{i,1}_u, v^{i,2}_u)$ is at  most $n_p-1$. 
        
        However, we note that there are unpopular players in the backbone in $A_{\sigma_0}$. 
        The games won by these players do not contribute to the score in $A_{\sigma_0}$ since the winners are unpopular. 
        But when the popular player $w^i$ wins the block, it beats all players (possibly unpopular) introduced after it and before $v^{i,2}_u$, and so the score increases.
        We claim that the score does not decrease when $w^i$ wins the block.
        Suppose not.
        Let $A'_{\sigma_0}$ denote the arborescence in which $w^i$ wins the block. 
        Note that in $A'_{\sigma_0}$, the unpopular player $v^{i,2}_u$ wins all its games after introduction.
        By extending $A'_{\sigma_0}$, we also obtain an arborescence in the whole seeding $\sigma$ in which $v^{i,2}_u$ wins all its games after introduction (with the same score as that of $A'_{\sigma_0}$ since $v^{i,2}_u$ is unpopular).
        Since the new arborescence has score strictly less than $\ell d + n_p - 1$ (due to the decrease), we have a contradiction to the fact that $\sigma$ has \vnar at least $\ell d + n_p -1$.
        
        The increase is denoted by $s^i$, and thus the score of $\sigma_0[v^{i,1}_p, v^{i,2}_u)$ when $w^i$ wins the block is at most the score given by $A_{\sigma_0}$ for the sub-seeding $\sigma_0[v^{i,1}_p, v^{i,2}_u)$ plus $s^i$, which is at most $s^i + n_p-1$. 
        Moreover the score given by $\sigma_0[v^{i,2}_u, *]$ is 0, since the unpopular player $v^{i,2}_u$ beats everybody.
    \end{proof}

The next proof is crucial to our proof of detecting sufficiently long paths in $D$. 
We do so by showing that there is a prefix with large score and no low scoring tagged blocks. 
More specifically, we show that there is {\it prefix} of $\sigma_0$, namely $\sigma_0[*, v^{k,1}_p)$, where $k$ is the first low scoring block, that must have a large score. 

\begin{lemma}
\label{lem:score-first-low-block}
Suppose that $\sigma_0[v^{k,1}_p, v^{k,2}_u]$ denotes the first low-scoring block in $A_{\sigma_0}$. 
Then, the score obtained in $A_{\sigma_0}$ due to the sub-seeding $\sigma_0[*,v^{k,1}_p)$ is at least $\ell d/2 - n_p$.
\end{lemma}
\begin{proof}
    We decompose $\sigma_0$ into sub-seedings $\sigma_0[*, v^{k,1}_p)$, $\sigma_0[v^{k,1}_p, v^{k,2}_u)$, and $\sigma_0[v^{k,2}_u, *]$.
    Recall that the total score of $\sigma_0$ is $\ell d+ n_p-1$. 
    Next, we will analyze the score of $\sigma_0$ given by the alternate event when $w^k$ wins the block.
    
    Consider the sequence of probabilistic outcomes where everything is identical to that in $A_{\sigma_0}$ up to the introduction of $w^k$ in the $k$th block, followed by $w^k$ winning the block.
    The resulting score is due to the score in the sub-seedings $\sigma_0[*, v^{k,1}_u)$, the $k$th block $\sigma_0[ v^{k,1}_u, v^{k,2}_u)$, and the remaining $\sigma_0[v^{k,2}_u, *]$.
    Due to \Cref{obs:max-score-of-tagged-block},\hide{ we know that when $w^k$ wins the block,} the score obtained from $\sigma_0[v^{k,1}_p,*]$ in the new event is at most $s^k + n_p$.
    Let $A'_{\sigma_0}$ denote the arborescence in which $w^k$ wins the block. 
    Note that in $A'_{\sigma_0}$, the unpopular player $v^{k,2}_u$ wins all its games after introduction.
    By extending $A'_{\sigma_0}$, we also obtain an arborescence in the whole seeding $\sigma$ in which $v^{k,2}_u$ wins all its games after introduction (with the same score as that of $A'_{\sigma_0}$ since $v^{k,2}_u$ is unpopular).
    Since the \vnar of $\sigma$ is at least $\ell d + n_p-1$, we have that 
    \begin{align*}
        \score(A_{\sigma_0[*, v^{k,1}_p)}) + s^k + n_p &= \score(A'_{\sigma_0[*, v^{k,1}_p)}) + s^k + n_p \\
        &\ge \score(A'_\sigma) \geq \ell d+ n_p-1.    
    \end{align*}
    Since the $k$th block is low scoring, we have that $s_k \le \ell d/2+n_p-1$.
    Hence, the score obtained by $A_{\sigma}$ in $\sigma[*, v^{k,1}_p)$ is at least $\ell d/2-n_p$ and the lemma is proved. 
\end{proof}

Now we are ready to obtain a prefix $\sigma_1$ of $\sigma_0$ (and thus, of $\sigma$ as well) that does not contain any low-scoring blocks.
If $\sigma_0$ does not contain any low-scoring blocks, then we set $\sigma_1=\sigma_0$.
Otherwise, we set $\sigma_1 = \sigma_0[*, v^{k,1}_p)$, where $k$ denotes the first low scoring tagged block.
From now on, we use $A_{\sigma_1}$ to denote the partial arborescence obtained from $A_{\sigma_0}$ by only considering the players in $\sigma_1$.
Applying \Cref{lem:score-first-low-block}, we have the following.

\begin{corollary}\label{cor:properties-of-prefix}
There is a prefix $\sigma_1$ of $\sigma_0$, such that the score obtained by the arborescence $A_{\sigma_1}$ is at least $\ell d/2 - n_p$. 
Moreover, if there is a tagged block $i$ in $\sigma_1$, then $s^i > \ell d/2 + n_p-1$. 
\end{corollary}

Consequently, what remains to be shown is how the existence of a prefix with a large score allows us to obtain a path of a certain length. 
\Cref{fig:mul-hoa-overview} illustrates a snapshot of the arborescence at this point in the proof.
\definecolor{fig-green}{HTML}{008000}
\definecolor{fig-orange}{HTML}{FFA500}
\begin{figure}
    \centering
    \includegraphics[width=\linewidth]{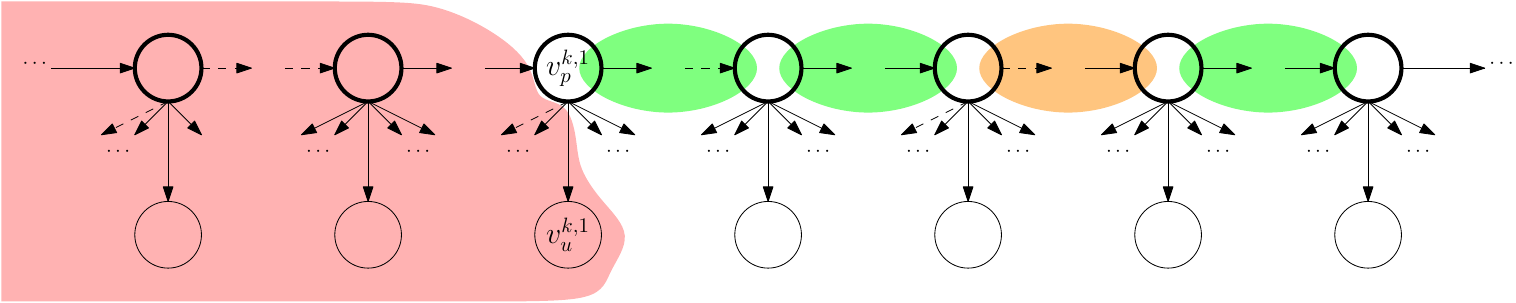}
    \caption{
    A figure illustrating the key ideas used thus far in the proof of \Cref{thm:multiplicative-approximation-NP-hardness}. Here we represent $A_\sigma$ using the same conventions as \Cref{fig:tagged-pair-example}.
   \hide{It contains $A_\sigma$, with the bold circles representing the popular backbone players, and the normal circles being their unpopular copies. The solid arrows represent the deterministic relation between these parents and children. The dashed arrows in the backbone represent potentially uncertain edges.} 
       The \textcolor{green}{{\bf green}} and \textcolor{fig-orange}{{\bf orange}} ovals denote clean blocks and high-scoring tagged blocks, respectively. The first low-scoring tagged block starts with $v^{k,1}_p$.
    The red area is the removed portion because of a low-scoring tagged block.
    \label{fig:mul-hoa-overview}
    }
\end{figure}

\paragraph{Computing a Long Path.}\label{para:path-computation}
Using \Cref{lem:clean-block}, we know that in clean blocks, there are deterministic paths connecting the two popular partners of the block with intermediate vertices from the block itself. By stitching together these paths, we can find a longer path in the graph $D$.
\hide{More generally, between any pair of popular players $(v^1_p, v^2_p)$ in the backbone who have their unpopular copies as children, there may exist some unpopular player in the sub-seeding who is also part of the backbone. 
As long as there exists a deterministic path from $v^2_p$ to $v^1_p$ with intermediate vertices corresponding to popular players introduced between $v^1_u$ and $v^2_u$, the blocks are clean and we are able to use this path, along with those obtained from clean blocks, to stitch together a longer path in the graph $D$.}
However, it is not clear how to deal with tagged blocks -- pairs of popular players who have their unpopular copies as children but do not have a deterministic path.

In the following lemma, we address this matter based on the fact that the score of $A_{\sigma_1}$ is at least $\ell d/2-n_p$. The critical idea in the proof is to pin-point a sub-seeding of $\sigma_1$ which does not have a tagged block, and the backbone does not have any unpopular players and has a large number of popular players who beat each of their unpopular copies. Such a sub-seeding implies a directed path, as shown earlier in \Cref{lem:backbone-pop-det-path}.

\begin{lemma}
\label{lem:obtaining-the-long-path}
For large enough $n_p$, there exists an $\epsilon'>0$ such that in polynomial time we can find in $D$ a path of length at least $n_p^{\epsilon'}$.
\end{lemma}

\begin{proof} 
In the following claim, we show that any existing tagged pair in $A_{\sigma_0}$ contains a ``large'' number of popular players introduced between the corresponding unpopular children. Note that in $\sigma_0$ every tagged block is a high-scoring block.
Also, recall that $v^{i,2}_u$ is the first unpopular child of $v^{i,2}_p$, and $v^{i,1}_u$ is the last unpopular child of $v^{i,1}_p$.

The first claim, \Cref{clm:intro-bet-tagged-players}, lower bounds the number of popular players in the backbone that are inside a tagged block, i.e, do not have any unpopular children. In the next claim, we are able to upper bound the total number of such players across the seeding. In doing so, in \Cref{clm:nice-sub-seeding}, we are able to show the existence of a sub-seeding that does not have a tagged block, has a large number of popular players in the backbone who have a large number of children, and no unpopular player is in the backbone. After that, using that sub-seeding, and by invoking \Cref{lem:clean-block}, we infer the existence of a deterministic path in $D$.

\begin{claim-inside-lemma}
    \label{lem:intro-bet-tagged-players}
    \label{clm:intro-bet-tagged-players}
    Between the two unpopular players $v^{i,1}_u$ and $v^{i,2}_u$ of any existing tagged pair $(v^{i,1}_p, v^{i,2}_p)$, at least $\ell/2$ popular players must be introduced.
\end{claim-inside-lemma}

\begin{proof}
        Assume for contradiction that at most $\ell/2-1$ popular players are introduced in $\sigma_1(v^{i,1}_u, v^{i,2}_u)$.
        Due to the definition of $v^{i,1}_u$, it follows that all other children of the backbone player $v^{i,1}_p$ introduced after $v^{i,1}_u$ must be popular. 
        Similarly, all children of the backbone player $v^{i,2}_p$ introduced before $v^{i,2}_u$ must be popular. 
        Additionally, due to the condition~\ref{prop:non-interleaving-repeated} of tagged pairs no popular player introduced in  $\sigma_1(v^{i,1}_u, v^{i,2}_u)$ can have an unpopular child. 
        Thus, the total score given by $A_{\sigma_1}$ in $\sigma_1(v^{i,1}_u, v^{i,2}_u)$ is at most $n_p-1$.
        
        Since the score increase due to $w^i$ is at least $\ell d/2+n_p$, and each popular player has only $d$ unpopular copies, we have that there is at least one unpopular player $y_u$ introduced in $\sigma_1(v^{i,1}_u, v^{i,2}_u)$ whose popular copy is not introduced in $\sigma_1(v^{i,1}_u, v^{i,2}_u)$.
        Consider the probabilistic event in which 
        \begin{itemize}[itemsep=0pt]
            \item $y_u$ wins all games after its introduction and then loses to $v^{i,2}_u$;
            \item $v^{i,2}_u$ wins all games after its introduction till the end.
        \end{itemize}
        Since $v^{i,2}_p$ is introduced before $v^{i,2}_u$ in $\sigma_1$, we have an arborescence $A'_{\sigma_1}$ that has score that is at least one less than that of $A_{\sigma_1}$ (due to the unpopular player $v^{i,2}_u$ winning).
        By extending $A'_{\sigma_1}$, we also obtain an arborescence in the whole seeding $\sigma$ in which $v^{i,2}_u$ wins all its games after introduction (with the same score as that of $A'_{\sigma_1}$ since $v^{i,2}_u$ is unpopular).
        Since the new arborescence has score strictly less than $\ell d + n_p - 1$ (due to the decrease), we have a contradiction to the fact that $\sigma$ has \vnar at least $\ell d + n_p -1$.
\end{proof}

We begin by showing the existence of a sub-seeding with some structural properties which make the path-finding process easy: one of them is the absence of tagged blocks.

\begin{claim-inside-lemma}
\label{clm:nice-sub-seeding}
There exists an $\varepsilon'>0$, and a sub-seeding of $\sigma_0$ s.t. 
\begin{enumerate}
    \item if a popular player beats an unpopular player as a child, then the latter is a copy of the former;
    \item there is no tagged block; and
    \item the score is at least $(n_p^{\varepsilon'} +1)d$.
\end{enumerate}
\end{claim-inside-lemma}
\begin{proof}
    We begin by noting that \Cref{clm:intro-bet-tagged-players} deals with popular players that do not have any unpopular children, such players can together contribute at most $n_p-1$ to the score. We observe that if there are strictly more than $n_p - \ell/2 $ such players in $\sigma_1$, then there can only be strictly fewer than $n_p - (n_p-\ell/2)=\ell/2$ popular players that have children in $A_{\sigma_1}$. Since, by \Cref{obs:children-are-copies}, the children of such players can only be their unpopular copies, so even if each one had the maximum number of children possible, that is $d$, the total score obtained in $A_{\sigma_1}$ is at most $(\ell/2-1)d + n_p-1$. Since, $d = n_p^2$, this is strictly smaller number than the score $\ell d/2 - n_p$, given by \Cref{cor:properties-of-prefix}. Hence, the number of popular players without unpopular children is at most $n_p - \ell/2$.
    
    Since these players make up the interior of a tagged block, each of which must have at least $\ell/2$ popular players, it follows that there can only be at most $(n_p - \ell/2)/(\ell/2)$ tagged blocks. Moreover, these many tagged blocks partition $\sigma_1$ into at most $(n_p - \ell/2)/(\ell/2)+1= 2(n_p - \ell/2)/\ell +1 $ (possibly empty) sub-seedings, called {\it clean chunks}, each of which do not have any tagged blocks. In other words, $\sigma_1$ may be viewed as alternating between clean chunks and tagged blocks, where a clean chunk may actually be empty or otherwise may contain several successive clean blocks. For example $\sigma_1$ may look as follows: $[C_1, T_1, C_2,T_2,C_3]$, where each of the $C_i$s and $T_i$s denote clean chunks and tagged blocks. 
    
    The simplest case where there is no tagged block in $\sigma_1$ reduces to the seeding being a single clean chunk and there is nothing more to prove. 
    
    Consider an arbitrary tagged block, denoted by $\sigma_1[v^{i,1}_p, v^{i,2}_u]$. 
    From this block, we are removing from consideration, the sub-seeding $\sigma_1(v^{i,1}_u, v^{i,2}_p)$, called the {\it unclean part} of this tagged block. 
    After removal, we only consider the following sub-seedings from this tagged block:
    $\sigma_1[v^{i,1}_p, v^{i,1}_u]$ and $\sigma_1[v^{i,2}_p, v^{i,2}_u]$.
    We note that the unclean part contains all of the unpopular players after $v^{i,1}_u$ that are inside this tagged block as well as the popular players that do not have any unpopular children. 
    Thus, each popular player from an unclean part can contribute by beating other popular players or at most one unpopular player in the backbone. Therefore, the total contribution to the score by the unclean parts across all the tagged blocks is at most $n_p$. Consequently, the total score in $\sigma_1$ outside the unclean parts is at least $\ell d/2 - n_p - n_p= \ell d/2 - 2n_p$. Moreover, we note that this score can be viewed as being split among clean chunks.
    
    As noted above, there may be at most $2(n_p - \ell/2)/\ell +1 $ clean chunks in $\sigma_1$. Moreover, after removal of the unclean parts, it looks as $[C_1, X^1_1, X^2_1, C_2, X^1_2, X^2_2, \ldots, X^1_{t-1}, X^2_{t-1}, C_t]$, where 
    $X^1_i$ and $X^2_i$ are the leftover clean pieces of the 
    $i$th tagged block. Hence, $[C_i, X^1_i]$ and $[X^2_i, C_{i+1}]$ are sub-seedings of $\sigma_1$. 
    Consequently, the scores from the leftover parts of the tagged blocks can be apportioned to the adjacent clean chunks. 
    
    Moreover, by an averaging argument, since the score of $A_{\sigma_1}$ (that is at least $\ell d/2 - 2n_p$) is split among $1+(n_p-\frac\ell2)$ 
    sub-seedings, there must be one with value at least the following, given that $\ell = (n_p)^\epsilon$. 
    \begin{align*}
        \frac{1}{1+(n_p-\frac{\ell}{2}) \frac2\ell} \left(\frac{\ell d}2-2n_p\right) &= \frac{\ell}{2n_p} \left(\frac{\ell d}2-2n_p\right) = \left(\frac{\ell^2}{4n_p}-\frac{\ell}{d}\right)d \\
        &= \left(\frac{n_p^{2\varepsilon}}{4n_p}-\frac{n_p^\varepsilon}{d}\right)d > (n_p^{\varepsilon'}+1)d.
    \end{align*}
    Note that there always exists some fixed constant $0 < \varepsilon' \le 1$ that makes the last inequality hold since $\varepsilon > 0.5$ and $n_p$ is sufficiently large.
    
    We conclude the argument by noting that each sub-seeding considered above \hide{obtained after removing the unclean part from a tagged block} satisfies the first two properties of the claim. The first one is due to the definition of $A_{\sigma}$ (\Cref{obs:children-are-copies}), and the second one is due to the fact that it is a sub-seeding obtained after removing the unclean parts. 
\end{proof}

The rest of the proof is to deal with  
the sub-seeding, denoted $\sigma_2$, obtained by the application of \Cref{clm:nice-sub-seeding} to $\sigma_1$. Let $A_{\sigma_2}$ denote the corresponding arborescence obtained from $A_{\sigma_1}$ by deleting vertices corresponding to players not in $\sigma_2$.
Observe that the lower-bound on score and the properties of $A_{\sigma_2}$ imply that there are at least $n_p^{\varepsilon'}+1$ popular players who have an unpopular copy of theirs as a child.
Every two such adjacent popular players form a clean pair as they satisfy the properties in \Cref{def:clean-pair}.
After repeated application of \Cref{lem:clean-block} 
on these adjacent pairs of popular players,
we obtain, in polynomial time, a path of length $n_p^{\varepsilon'}$ in the graph $D$.
Thus, the lemma is proved. 
\end{proof}

Recall that the input digraph $D$ is assumed to contain a Hamiltonian path, and it is \nph to find a path of length $(n_p-1)/n_p^{1-\varepsilon''} \le n_p^{\varepsilon'}$ for an appropriately chosen $\varepsilon''$ given $\varepsilon'$ (\Cref{prop:source-hardness}).
Combining \Cref{lem:obtaining-the-long-path} with the fact that the \vnar of $\sigma$ is at least $\ell d + n_p-1$ where $\ell = n_p^\varepsilon$ with $0.5 < \varepsilon \le 1$, we have the following statement on the hardness of approximation.

\thma*

\section{Algorithms for the \vnar}\label{dynamic}
\label{sec:finding-apx-seeding}

We first show a non-adaptive algorithm for a special case where each uncertain edge in the strength graph is either between two popular players, or two unpopular players. Thus, any edge between a popular and an unpopular player is deterministic. We establish a lower bound on \vnar in terms of the chromatic number of the graph $G_p=(V_p,E_p)$, which is a subgraph of the strength graph where $V_p$ is the set of popular players and $E_p$ is the set of uncertain edges incident on $V_p$.

\chromaticrestricted*
\begin{proof}
As before, we say a player $p$ {\it beats} $p'$ if $p$ beats $p'$ deterministically. We describe our algorithm for the required seeding.
\begin{enumerate}[]
    \item Let $C_1, \ldots, C_k$ be a proper $k$-coloring of $G_p$.
    Note that there is no edge between two vertices in any set/color class $C_i$.
    In particular, there is no uncertain edge between any two vertices in $C_i$.
    
    \item For each $i \in [k]$, let $P_i$ be a Hamiltonian path in the strength graph restricted to the vertices of $C_i$; note that such a path exists (and can be found in polynomial time) because the considered graph is a complete directed graph. 
    Let $f_i$ and $\ell_i$ denote the first and last vertices of $P_i$, respectively.
    We will process the color classes in the increasing order of their label.
    
    \item\label{item:alg-chr-num-step3} We construct the seeding by introducing the player $\ell_1$ first, and then sequentially introducing each unpopular player\hide{ deterministically} beaten by $\ell_1$ (in an arbitrary order), one at a time. 
    After this, we introduce the next popular player along the Hamiltonian path, and then each unpopular player\hide{ deterministically} beaten by it (who has not played before).
    Continuing in this manner we place all the vertices in $C_1$, and all the unpopular players\hide{ deterministically} beaten by at least one player in $P_1$.
    Observe that the champion\hide{ of the tournament} after all players in $C_1$ have played is $f_1$.
    
    \item We consider the relationship between $\ell_2$ and $f_1$.
    If $\ell_2$ beats $f_1$\hide{ deterministically}, we proceed by introducing $\ell_2$, then each unpopular player\hide{ deterministically} beaten by it, and so on following Step \ref{item:alg-chr-num-step3}. 
    Otherwise, either $f_1$ beats $\ell_2$\hide{ deterministically} or there is an uncertain edge between $f_1$ and $\ell_2$. 
    Let $X$ be the set of unpopular players that are\hide{ deterministically} beaten by $\ell_2$.
    One of the following two scenarios must occur:
    \begin{enumerate}[wide=10pt]
       \item\label{item:alg-chr-num-step4} There exists an unpopular player $v \in X$ that beats $f_1$\hide{ deterministically}. 
       Then, we first introduce $v$ in the seeding, followed by $\ell_2$, and then the remaining vertices in $X$ (if any); and then move along the path $P_2$ following Step \ref{item:alg-chr-num-step3}.
       \item All players in $X$ are\hide{ deterministically} beaten by $f_1$. 
       \hide{In this case}Then , we first introduce all vertices in $X$, followed by $\ell_2$.
       If $\ell_2$ wins, we move along the path $P_2$ following Step \ref{item:alg-chr-num-step3}; else, we repeat the same arguments while comparing $f_1$ with the next player along the path $P_2$, and continue the above process on the remaining vertices in this color class, and then on the remaining color classes.
    \end{enumerate}
\end{enumerate}
Next, we bound the value of the constructed seeding.

\setcounter{personalcounterone}{\value{theorem}}
\setcounter{personalcounterzero}{\value{section}}
\setcounter{theorem}{\thepersonalcountertwo}
\setcounter{section}{\thepersonalcounterthree}
\begin{lemma-inside-theorem} 
    \label{clm:chrm-score-bound}
    The constructed seeding has \vnar at least $(n_p+n_u-1)-(k-1)$.
\end{lemma-inside-theorem}
\setcounter{theorem}{\thepersonalcounterone}
\setcounter{section}{\thepersonalcounterzero}
 \begin{proof}
    Recall that the maximum value that can be achieved by any seeding is bounded by $|n_p|+|n_u|-1$.
    Thus, it is sufficient to show that the number of games in the tournament that are won by players in $U$ is at most $k-1$.
    Such games are precisely the ones described in Step \ref{item:alg-chr-num-step4}.
    By construction, there is at most one such game while processing each color class $C_i$, and there is no such game when processing color class $C_1$.
    The number of color classes is $k$, and we have the desired lemma.
      \end{proof}
      This completes the proof of the theorem.
\end{proof}

Recall that we supposed that we have an instance in which for each uncertain edge, the end points are either both popular or both unpopular.
\hide{Let $OPT$ denote the maximum $\vnar$ witnessed by any seeding in this instance.} Let $\vnar^* $ denote the optimal \vnar score in the given instance. 
Note that the proof of \Cref{thm:chromatic1} yields the following result since each step can be executed in polynomial time.

\begin{observation}
    \label{obs:coloring-to-seeding}
    Given a proper coloring of $G_p$ with $\hat{k}$ colors, a seeding with \vnar at least $(n_p+n_u-1)-(\hat{k}-1) \ge \vnar^* - (\hat{k}-1)$ can be computed in polynomial time.
\end{observation}

Let $\vcprob$ denote the size of a minimum vertex cover in $G_p$.
It is known that a vertex cover in $G_p$ of size at most $2\vcprob$ can be computed in polynomial time.
Let $S$ denote such a vertex cover.
From $S$, we obtain a $(|S|+1)$-proper coloring of $G_p$ as follows: $|S|$ many colors for each vertex in $S$ and 1 color for the vertices in $V_p \setminus S$.
Applying \Cref{obs:coloring-to-seeding}, we have the following.

\begin{theorem}
    \label{thm:additive-approx-vc}
    Suppose that for each uncertain edge in the strength graph, the end points are either both popular or both unpopular.
    A seeding with \vnar at least $n_p+n_u-1 - (2\vcprob-1) \ge \vnar^* -(2\vcprob-1)$ can be computed in polynomial time.
\end{theorem}

\paragraph{Beyond Polynomial Time.} More generally, consider any graph class $\mathcal{C}$ for which there is an algorithm, be it exact-exponential, \FPT\ or \XP\ etc,  to compute the chromatic decomposition with time complexity $T_{\C}$. Then, we can use that algorithm in conjunction with our construction to create a seeding with a guaranteed score.

\begin{theorem}
Suppose that a proper coloring of $G_p$ that uses $k$ colors can be computed in time $T(G_p)$. 
Then, there exists an algorithm that computes a seeding with \vnar at least $(n_p \!+\! n_u \!-\!1)\!-\!(k\!-\!1)$ in time $T(G_p) \!+\! n^{\bigoh(1)}$, where $n$ is the number of players in the instance.
\end{theorem}

Specifically, for classes for which an optimal chromatic decomposition can be found efficiently, we have an algorithm with the same time complexity and additive polynomial time that can compute a seeding of value at least $(n_p + n_u \!-\! 1)-(\chi(G_p)\!-\!1)$. 
Such classes include perfect graphs and bipartite graphs, where the problem is polynomial-time solvable; bounded tree-width graphs, where it is FPT w.r.t. $k+ tw$~\cite{DBLP:books/sp/CyganFKLMPPS15}; and $P_5$-free graphs, where it is XP w.r.t. $k$~\cite{DBLP:journals/algorithmica/HoangKLSS10}.

From now on, we consider the problem without any restrictions on the end points of uncertain edges.
We note that the instance in \Cref{fig:example-with-pu-prob-edges} exhibits an example where the guarantees of \Cref{thm:chromatic} do not hold.
Observe that in that instance we have $n_p=2$, $n_u=8$,$\chi(G_p)=2$, and although the bound given by \Cref{thm:chromatic} is 8, the \vnar is only 7, witnessed by the seeding $\langle a, x_1, x_2, x_3, x_4, y_1, y_2, b, y_3, y_4 \rangle$.

\subsection{Adaptive Algorithm}
\label{subsec:adaptive-seeding}

Recall that an algorithm is \emph{adaptive} if it constructs a seeding adaptively, depending on the outcomes of earlier matches.

\thmb*
\begin{proof}
The first three steps of our algorithm (\Cref{alg:adaptive_algorithm}) are the same as that in the proof of \Cref{thm:chromatic1}.

\DontPrintSemicolon 
\begin{algorithm}[t]
\caption{An Adaptive Algorithm}
\label{alg:adaptive_algorithm}
\SetKwInOut{Input}{Input}
\SetKwInOut{Output}{Output}
\SetKwFunction{Match}{Match}

\newcommand{\currpop}{\mathrm{curr\_pop}}

\Input{Sets $\mathcal{P}, \mathcal{U}$, graph $G$, and proper coloring $(C_1, \dots, C_k)$ of $G_p$}

\LinesNumbered
\tcp{Preprocessing}
For each $i \in [k]$, let $P_i = (a^i_1, \dots, a^i_{n_i})$ be a Hamiltonian path in $G[C_i]$, where $a^i_j$ beats $a^i_{j+1}$\;
$\currpop \gets \bot$ \quad \quad \tcp{Assume all players beat dummy player $\bot$}
\For{$i \gets 1$ \KwTo $k$}{
    \For{$j \gets n_i$ \KwTo $1$}{
        \If{$a^i_j$ beats $\currpop$}{
            $a^i_j$ plays $\currpop$\;
            $\currpop \gets a^i_j$\;
        }
        \ForEach{unpopular player $u$ beaten by $a^i_j$}{
            $u$ plays $\currpop$\;
            \If{the winner is $u$ \quad \quad \tcp{Only possible if $\currpop \neq a^i_j$}}{
                $u$ plays $a^i_j$\;
                $\currpop \gets a^i_j$\;
            }
        }
        \If{$\currpop \ne a^i_j$}{
            $a^i_j$ plays $\currpop$\;
            \If{the winner is $a^i_j$}{
                $\currpop \gets a^i_j$\;
            }
        }
    }
}
\end{algorithm}

\begin{enumerate}[]
    \item Let $C_1, \ldots, C_k$ be the proper $k$-coloring of $G_p$.
    Note that there is no edge between two vertices in any set/color class $C_i$.
    In particular, there is no uncertain edge between any two vertices in $C_i$.
    
    \item For each $i \in [k]$, let $P_i$ be a Hamiltonian path in the strength graph restricted to the vertices of $C_i$; note that such a path exists (and can be found in polynomial time) because the considered graph is a complete directed graph. 
    Let $f_i$ and $\ell_i$ denote the first and last vertices of $P_i$, respectively.
    We will process the color classes in the increasing order of their label.
    
    \item\label{item:alg-ada-chr-num-step3} We construct the seeding by introducing the player $\ell_1$ first, and then sequentially introducing each unpopular player\hide{ deterministically} beaten by $\ell_1$ (in an arbitrary order), one at a time. 
    After this, we introduce the next popular player along the Hamiltonian path, and then each unpopular player\hide{ deterministically} beaten by it (who has not played before).
    Continuing in this manner we place all the vertices in $C_1$, and all the unpopular players\hide{ deterministically} beaten by at least one player in $P_1$.
    Observe that the champion\hide{ of the tournament} after all players in $C_1$ have played is $f_1$.

\item \label{item:alg-ada-chr-num-step4}
If $\ell_2$ beats $f_1$\hide{ deterministically}, we proceed by introducing $\ell_2$, then each unpopular player\hide{ deterministically} beaten by $\ell_2$, and so on following Step \ref{item:alg-chr-num-step3}. 
Otherwise, either $f_1$ beats $\ell_2$\hide{ deterministically} or there is an uncertain edge between $f_1$ and $\ell_2$. 
Let $X$ be the set of unpopular players that are\hide{ deterministically} beaten by $\ell_2$.
Note that there is no player in $X$ that\hide{ deterministically} loses to $f_1$.
One of the following two scenarios must occur:
\begin{enumerate}[wide=10pt]
   \item\label{item:alg-ada-chr-num-step4a} There exists an unpopular player $v$ in $X$ that beats $f_1$\hide{ deterministically}. 
   In this case, we first introduce $v$ in the seeding, followed by $\ell_2$, and then the remaining vertices in $X$ (if any), and then move along the path $P_2$ following Step \ref{item:alg-chr-num-step3}.
   \item All players in $X$ have an uncertain edge with $f_1$. 
   We introduce the players in $X$ (in an arbitrary order), and stop if any such player beats $f_1$.
   If one of the players beat $f_1$, then we introduce $\ell_2$, and then the remaining vertices in $X$ (if any), and then move along the path $P_2$ following Step \ref{item:alg-chr-num-step3}.
   Else all players in $X$ have been introduced and $f_1$ is the champion, and we introduce $\ell_2$.
   If $\ell_2$ wins, we move along the path $P_2$ following Step \ref{item:alg-chr-num-step3}; else, we repeat the same arguments while comparing $f_1$ with the next player along the path $P_2$, and continue the above process on the remaining vertices in this color class, and then on the remaining color classes.
\end{enumerate}
\end{enumerate}
The algorithm obtains \vnar at least $(|n_p|+|n_u|-1)-(k-1)$, via a proof identical to \Cref{clm:chrm-score-bound} in \Cref{thm:chromatic1}.
\end{proof} 

\begin{remark}
    Consider the execution of \Cref{alg:adaptive_algorithm} on the instance shown in Figure~\ref{fig:example-with-pu-prob-edges}, with the coloring $C_1 =\{a\}$ and $C_2 = \{b\}$. 
    The algorithm first introduces $a$, who defeats the unpopular players $x_1, \dots, x_4$ deterministically. 
    When transitioning to player $b$, the algorithm does not match $a$ and $b$ directly due to the uncertain edge between them. 
    Instead, it uses the unpopular players $y_1, \dots, y_4$ (who are deterministically beaten by $b$) to challenge $a$. 
    If $a$ defeats all $y_i$, $a$ remains the champion. 
    If $a$ loses to some $y_i$, then $b$ immediately enters, defeats $y_i$ deterministically, becomes the champion, and proceeds to defeat the remaining $y$'s. 
    In either case, the only potential loss in value comes from a single match where a popular player ($a$) loses to an unpopular player ($y_i$).
    Thus, \Cref{alg:adaptive_algorithm} obtains a seeding with \vnar 8.
\end{remark}
\section{Conclusion}
In this work, we introduced and studied the \vnar objective for \ctc tournaments. \vnar provides a robust guarantee, ensuring a minimum achievable value regardless of the outcomes of uncertain games. Our results include polynomial-time algorithms and nearly matching hardness of approximation results.

\begin{sloppypar}
Coping with this hardness is a direction for future work. Potential approaches include exponential-time algorithms, exploration of parameters for fixed-parameter tractable algorithms, and algorithms for specific input restrictions.
Beyond \vnar, related objective functions\hide{ for Challenge-the-Champ tournaments} in tournaments with probabilistic outcomes include: maximizing expected value 
and maximizing the likelihood of achieving a certain value.
Finally, the scoring function too can be generalized to allow player and round dependencies.  
\end{sloppypar}

\section*{Acknowledgments}
UB and JC acknowledge support from the Department of Atomic Energy, Government  of India, under project no. RTI4014. Part of this work was done while JC was a visiting fellow in TIFR.

\bibliographystyle{alpha} 
\bibliography{tournaments}

@inproceedings{chaudhary2024make,
  title={How to make knockout tournaments more popular?},
  author={Chaudhary, Juhi and Molter, Hendrik and Zehavi, Meirav},
  booktitle={Proceedings of the AAAI Conference on Artificial Intelligence},
  volume={38},
  number={9},
  pages={9582--9589},
  year={2024}
}

@article{dagaev2018competitive,
  title={Competitive intensity and quality maximizing seedings in knock-out tournaments},
  author={Dagaev, Dmitry and Suzdaltsev, Alex},
  journal={Journal of Combinatorial Optimization},
  volume={35},
  number={1},
  pages={170--188},
  year={2018},
  publisher={Springer}
}

@inproceedings{bhaskar2025maximizing,
  title={Maximizing Value in Challenge the Champ Tournaments},
  author={Bhaskar, Umang and Chaudhary, Juhi and Dey, Palash},
  booktitle={Proceedings of the 24th International Conference on Autonomous Agents and Multiagent Systems},
  pages={307--315},
  year={2025}
}

@inproceedings{chaudhary2025adaptive,
  title={Adaptive Manipulation for Coalitions in Knockout Tournaments},
  author={Chaudhary, Juhi and Molter, Hendrik and Zehavi, Meirav},
  booktitle={Proceedings of the AAAI Conference on Artificial Intelligence},
  volume={39},
  number={13},
  pages={13700--13708},
  year={2025}
}

@article{suksompong2016scheduling,
  title={Scheduling asynchronous round-robin tournaments},
  author={Suksompong, Warut},
  journal={Operations Research Letters},
  volume={44},
  number={1},
  pages={96--100},
  year={2016},
  publisher={Elsevier}
}

@inbook{williams_moulin_2016, place={Cambridge}, title={Knockout Tournaments}, DOI={10.1017/CBO9781107446984.020}, booktitle={Handbook of Computational Social Choice}, publisher={Cambridge University Press}, author={Williams, Virginia Vassilevska}, editor={Brandt, Felix and Conitzer, Vincent and Endriss, Ulle and Lang, Jérôme and Procaccia, Ariel D.Editors}, year={2016}, pages={453–474}}

@inproceedings{russell2009manipulating,
  title={Manipulating tournaments in cup and round robin competitions},
  author={Russell, Tyrel and Walsh, Toby},
  booktitle={Algorithmic Decision Theory: First International Conference, ADT 2009, Venice, Italy, October 20-23, 2009. Proceedings 1},
  pages={26--37},
  year={2009},
  organization={Springer}
}

@inproceedings{suksompong2021tournaments,
  title={Tournaments in Computational Social Choice: Recent Developments},
  author={Suksompong, Warut},
  booktitle={IJCAI},
  pages={4611--4618},
  year={2021}
}

@inproceedings{aziz2014fixing,
  title={Fixing a balanced knockout tournament},
  author={Aziz, Haris and Gaspers, Serge and Mackenzie, Simon and Mattei, Nicholas and Stursberg, Paul and Walsh, Toby},
  booktitle={Proceedings of the AAAI Conference on Artificial Intelligence},
  volume={28},
  number={1},
  year={2014}
}

@article{ArlegiD20,
  author       = {Ritxar Arlegi and
                  Dinko Dimitrov},
  title        = {Fair elimination-type competitions},
  journal      = {Eur. J. Oper. Res.},
  volume       = {287},
  number       = {2},
  pages        = {528--535},
  year         = {2020}
}

@book{knockout,
title={{Knockout Tournaments}},  
booktitle={Handbook of computational social choice},
  author={Brandt, Felix and Conitzer, Vincent and Endriss, Ulle and Lang, J{\'e}r{\^o}me and Procaccia, Ariel D},
  year={2016},
  publisher={Cambridge University Press}
}

@inproceedings{efremenko2025tournament,
  title={Tournament Robustness via Redundancy},
  author={Efremenko, Klim and Molter, Hendrik and Zehavi, Meirav},
  booktitle={Proceedings of the 26th ACM Conference on Economics and Computation},
  pages={66--85},
  year={2025}
}

@book{laslier1997tournament,
  title={Tournament solutions and majority voting},
  author={Laslier, Jean-Fran{\c{c}}ois},
  volume={7},
  year={1997},
  publisher={Springer}
}

@misc{rosen1985prizes,
  title={Prizes and incentives in elimination tournaments},
  author={Rosen, Sherwin},
  year={1985},
  publisher={National Bureau of Economic Research Cambridge, Mass., USA}
}

@book{buchanan1980toward,
  title={Toward a Theory of the Rent-seeking Society},
  author={Buchanan, J.M. and Tollison, R.D. and Tullock, G.},
  number={no. 4},
  isbn={9780890960905},
  lccn={79005276},
  series={Texas A \& M University economics series},
  url={https://books.google.co.in/books?id=DC-8AAAAIAAJ},
  year={1980},
  publisher={Texas A \& M University}
}

@inproceedings{vu2009complexity,
  title={On the complexity of schedule control problems for knockout tournaments.},
  author={Vu, Thuc and Altman, Alon and Shoham, Yoav},
  booktitle={AAMAS (1)},
  pages={225--232},
  year={2009}
}

@article{VuShoham2011, 
title = {Fair seeding in knockout tournaments}, 
author = {Vu, Thuc and Shoham, Yoav}, 
journal = {ACM Transactions on Intelligent Systems and Technology}, volume = {3}, 
number = {1}, 
year = {2011}}

@inproceedings{ChatterjeeIbsenJensenTkadlec2016,
  author       = {Krishnendu Chatterjee and
                  Rasmus Ibsen{-}Jensen and
                  Josef Tkadlec},
  editor       = {Subbarao Kambhampati},
  title        = {Robust Draws in Balanced Knockout Tournaments},
  booktitle    = {Proceedings of the Twenty-Fifth International Joint Conference on
                  Artificial Intelligence, {IJCAI} 2016},
  pages        = {172--179},
  publisher    = {{IJCAI/AAAI} Press},
  year         = {2016},
  url          = {http://www.ijcai.org/Abstract/16/032},
  bibsource    = {dblp computer science bibliography, https://dblp.org}
}

@inproceedings{blavzej2024controlling,
  title={On Controlling Knockout Tournaments Without Perfect Information},
  author={Bla{\v{z}}ej, V{\'a}clav and Gupta, Sushmita and Ramanujan, MS and Strulo, Peter},
  booktitle={19th International Symposium on Parameterized and Exact Computation},
  year={2024}
}

@article{chaudhary2025parameterized,
  title={Parameterized analysis of bribery in challenge the champ tournaments},
  author={Chaudhary, Juhi and Molter, Hendrik and Zehavi, Meirav},
  journal={Journal of Artificial Intelligence Research},
  volume={83},
  year={2025}
}

@article{YangD21,
  title={Weak transitivity and agenda control for extended stepladder tournaments},
  author={Yang, Yongjie and Dimitrov, Dinko},
  journal={Economic Theory Bulletin},
  volume={9},
  pages={27--37},
  year={2021},
  publisher={Springer}
}

@article{mattei2015complexity,
  title={On the complexity of bribery and manipulation in tournaments with uncertain information},
  author={Mattei, Nicholas and Goldsmith, Judy and Klapper, Andrew and Mundhenk, Martin},
  journal={Journal of Applied Logic},
  volume={13},
  number={4},
  pages={557--581},
  year={2015},
  publisher={Elsevier}
}

@article{aziz2018fixing,
  title={Fixing balanced knockout and double elimination tournaments},
  author={Aziz, Haris and Gaspers, Serge and Mackenzie, Simon and Mattei, Nicholas and Stursberg, Paul and Walsh, Toby},
  journal={Artificial Intelligence},
  volume={262},
  pages={1--14},
  year={2018},
  publisher={Elsevier}
}

@article{stanton2013structure,
  title={The structure, efficacy, and manipulation of double-elimination tournaments},
  author={Stanton, Isabelle and Williams, Virginia Vassilevska},
  journal={Journal of Quantitative Analysis in Sports},
  volume={9},
  number={4},
  pages={319--335},
  year={2013},
  publisher={De Gruyter}
}

@article{sauer2024improving,
  title={Improving ranking quality and fairness in Swiss-system chess tournaments},
  author={Sauer, Pascal and Cseh, {\'A}gnes and Lenzner, Pascal},
  journal={Journal of Quantitative Analysis in Sports},
  volume={20},
  number={2},
  pages={127--146},
  year={2024},
  publisher={De Gruyter}
}

@article{sziklai2022efficacy,
  title={The efficacy of tournament designs},
  author={Sziklai, Bal{\'a}zs R and Bir{\'o}, P{\'e}ter and Csat{\'o}, L{\'a}szl{\'o}},
  journal={Computers \& Operations Research},
  volume={144},
  pages={105821},
  year={2022},
  publisher={Elsevier}
}

@inproceedings{gupta2024exercise,
  title={An exercise in tournament design: When some matches must be scheduled},
  author={Gupta, Sushmita and Sridharan, Ramanujan and Strulo, Peter},
  booktitle={Proceedings of the AAAI Conference on Artificial Intelligence},
  volume={38},
  number={9},
  pages={9749--9756},
  year={2024}
}

@inproceedings{BlazejG0S24,
  author       = {V{\'{a}}clav Blazej and
                  Sushmita Gupta and
                  M. S. Ramanujan and
                  Peter Strulo},
  editor       = {{\'{E}}douard Bonnet and
                  Pawel Rzazewski},
  title        = {On Controlling Knockout Tournaments Without Perfect Information},
  booktitle    = {19th International Symposium on Parameterized and Exact Computation,
                  {IPEC} 2024},
  series       = {LIPIcs},
  volume       = {321},
  pages        = {7:1--7:15},
  publisher    = {Schloss Dagstuhl - Leibniz-Zentrum f{\"{u}}r Informatik},
  year         = {2024}
}

@inproceedings{DBLP:conf/icalp/BjorklundHK04,
  author       = {Andreas Bj{\"{o}}rklund and
                  Thore Husfeldt and
                  Sanjeev Khanna},
  editor       = {Josep D{\'{\i}}az and
                  Juhani Karhum{\"{a}}ki and
                  Arto Lepist{\"{o}} and
                  Donald Sannella},
  title        = {Approximating Longest Directed Paths and Cycles},
  booktitle    = {Automata, Languages and Programming: 31st International Colloquium,
                  {ICALP} 2004},
  series       = {Lecture Notes in Computer Science},
  volume       = {3142},
  pages        = {222--233},
  publisher    = {Springer},
  year         = {2004},
  doi          = {10.1007/978-3-540-27836-8\_21},
  bibsource    = {dblp computer science bibliography, https://dblp.org}
}

@article{DBLP:journals/jcss/ImpagliazzoP01,
  author       = {Russell Impagliazzo and
                  Ramamohan Paturi},
  title        = {On the Complexity of k-SAT},
  journal      = {J. Comput. Syst. Sci.},
  volume       = {62},
  number       = {2},
  pages        = {367--375},
  year         = {2001},
  doi          = {10.1006/JCSS.2000.1727},
  bibsource    = {dblp computer science bibliography, https://dblp.org}
}

@article{DBLP:journals/algorithmica/HoangKLSS10,
  author       = {Ch{\'{\i}}nh T. Ho{\`{a}}ng and
                  Marcin Kaminski and
                  Vadim V. Lozin and
                  Joe Sawada and
                  Xiao Shu},
  title        = {Deciding \emph{k}-Colorability of \emph{P}\({}_{\mbox{5}}\)-Free Graphs
                  in Polynomial Time},
  journal      = {Algorithmica},
  volume       = {57},
  number       = {1},
  pages        = {74--81},
  year         = {2010},
  doi          = {10.1007/S00453-008-9197-8},
  bibsource    = {dblp computer science bibliography, https://dblp.org}
}

@book{DBLP:books/sp/CyganFKLMPPS15,
  author       = {Marek Cygan and
                  Fedor V. Fomin and
                  Lukasz Kowalik and
                  Daniel Lokshtanov and
                  D{\'{a}}niel Marx and
                  Marcin Pilipczuk and
                  Michal Pilipczuk and
                  Saket Saurabh},
  title        = {Parameterized Algorithms},
  publisher    = {Springer},
  year         = {2015},
  doi          = {10.1007/978-3-319-21275-3},
  isbn         = {978-3-319-21274-6},
  bibsource    = {dblp computer science bibliography, https://dblp.org}
}

@book{DBLP:books/daglib/0022205,
  author       = {J{\o}rgen Bang{-}Jensen and
                  Gregory Z. Gutin},
  title        = {Digraphs - Theory, Algorithms and Applications, Second Edition},
  series       = {Springer Monographs in Mathematics},
  publisher    = {Springer},
  year         = {2009},
  isbn         = {978-1-84800-997-4},
  bibsource    = {dblp computer science bibliography, https://dblp.org}
}

@article{KhannaLS00,
  author       = {Sanjeev Khanna and
                  Nathan Linial and
                  Shmuel Safra},
  title        = {On the Hardness of Approximating the Chromatic Number},
  journal      = {Comb.},
  volume       = {20},
  number       = {3},
  pages        = {393--415},
  year         = {2000}
}

@article{BjorklundHK09,
  author       = {Andreas Bj{\"{o}}rklund and
                  Thore Husfeldt and
                  Mikko Koivisto},
  title        = {Set Partitioning via Inclusion-Exclusion},
  journal      = {{SIAM} J. Comput.},
  volume       = {39},
  number       = {2},
  pages        = {546--563},
  year         = {2009}
}

\end{document}